\documentclass[a4paper,12pt]{article}

\usepackage{amsmath,amssymb,amsthm,amsfonts}
\usepackage{color}
\usepackage{color,graphicx}

\usepackage[bookmarks=false]{hyperref}

\theoremstyle{plain}
\newtheorem{prop}{Proposition}[section]
\newtheorem{theorem}[prop] {Theorem}

\theoremstyle{definition}
\newtheorem{lemma}[prop]{Lemma}

\newtheorem{corollary}[prop]{Corollary}
\newtheorem{conjecture}[prop]{Conjecture}

\newcounter{acount}
\newtheorem{assumption}[acount]{Assumption}

\theoremstyle{remark}
\newtheorem*{remark}{Remark}

\newcommand{\N}{\mathbb{N}}
\newcommand{\R}{\mathbb{R}}
\newcommand{\Z}{\mathbb{Z}}

\newcommand{\dd}{\mathrm{d}} 
\newcommand{\eps}{\epsilon}

\renewcommand{\eps}{\varepsilon}

\newcommand{\vect}[1]{\boldsymbol{#1}}

\DeclareMathOperator{\const}{const}
\DeclareMathOperator{\dist}{dist}

\newcommand{\rmay}{R^\mathsf{May}} 
\newcommand{\rvir}{R^\mathsf{vir}}
\newcommand{\may}{\mathsf{May}}
\newcommand{\vir}{\mathsf{vir}}
\newcommand{\sat}{\mathrm{sat}}

\newcommand{\cl}{\mathrm{cl}}
 \newcommand{\conn}{\mathrm{conn}}

\title{Mayer and virial series at low temperature}
\author{Sabine Jansen\footnote{Weierstrass Institute 
for Applied Analysis and Stochastics, Leibniz Institute in Forschungsverbund Berlin e.V., Mohrenstr. 39, 
10117 Berlin, Germany. E-mail: \texttt{jansen@wias-berlin.de}. Phone: +49 30 20372 466.}}

\date{5 December 2011}

\begin{document}

\maketitle
\begin{abstract}
	We analyze the Mayer and virial series (pressure as a function of the activity resp. the  
		density) for a classical system of particles in continuous 
	configuration space at low temperature. Particles interact via a finite range 
	potential with an attractive tail.  
	We propose physical interpretations of the Mayer and virial series' radii of convergence, 
	valid independently of the question of phase transition: the Mayer radius
	corresponds to a fast increase  
	from very small to finite density, and the virial radius corresponds to 
	a cross-over from monatomic to polyatomic gas. Our results are consistent 
	with the Lee-Yang theorem for lattice gases and with
	 the continuum Widom-Rowlinson model.
\end{abstract}


\section{Introduction}

The present work started from a seeming contradiction between 
results on cluster size distributions at low temperature and low density~\cite{jkm} and 
predictions from the Mayer activity expansions. It turned out that not only is there no contradiction, 
but moreover the interplay between the two different approaches considerably helps
 the physical interpretation of the classical expansions. 

The seeming contradiction is the following. Consider a classical system of particles, interacting via a stable potential with
 an attractive tail. As is well-known from the theory of Mayer expansions (see e.g., the classical textbook~\cite{ruelle-book}), 
at low density, the system behaves approximately like an ideal gas, suggesting that particles move more or less independently  
and are typically far away from each other. But~\cite{jkm} showed that when both the density and temperature are small, 
there can be regimes where particles form small compounds -- the system could behave, for example, like a diatomic gas. 
Such a behavior is, in fact, well established for quantum Coulomb systems~\cite{bm,cly,fefferman}. \\

In order to fit both pictures together, we investigate the temperature dependence of the 
Mayer and virial series. The temperature dependence confirms the intuitive relation between 
Mayer's series and Frenkel-Band theory of physical clusters, as exposed in~\cite[Chapter 5]{hill1}: 
the Mayer series $\beta p = \sum b_k z^k$ looks like the pressure of an ideal mixture of size $k$ components 
with respective activities $b_k z^k$. But unlike physical activities, the coefficients $b_k$ can be negative. 
The gas is therefore, at best, an \emph{exact} ideal mixture of \emph{fictitious} objects, ``\emph{mathematical}'' clusters. On the 
other hand, at low density, it is tempting to consider the system as an \emph{approximately} ideal mixture of 
``\emph{physical}'' clusters, groups of particles close in configuration space. 

Each physical cluster comes with a partition 
function over internal degrees of freedom. 
At low temperatures, it is natural to approximate the internal partition function as $\exp(-\beta E_k)$, with 
$E_k$ a ground state energy, and we expect
\begin{equation}
  \beta p \approx \sum_k z^k \exp(-\beta E_k). \label{eq:intro-app}
\end{equation}
We prove that at low temperatures, the Mayer series coefficients $b_k(\beta)$ indeed behave as 
$\exp(-\beta (E_k +o(1)))$ (Theorem~\ref{thm:maylt}), so that the approximation described above matches the exact Mayer series. 
As a consequence, we can easily understand the formation of compounds: if $\beta \to \infty$ 
at fixed chemical potential $\mu$, we have to maximize $(k\mu - E_k)$ in order to see which 
$k$ gives the dominant contribution. In particular, even when the Mayer series converges, 
at low temperatures the main contribution does not necessarily come from $k=1$. \\

With the approximate formula~\eqref{eq:intro-app} in mind, we prove several results on Mayer and 
virial series and low temperature statistical physics, stated 
in Sect.~\ref{sec:results}; the proofs are given in Sects~4--7.  The main results for the pressure-activity 
expansion are the following: Theorem~\ref{thm:maylt} relates the Mayer coefficients $b_k(\beta)$ to the 
ground state energies $E_k$, justifying Eq.~\eqref{eq:intro-app} as described above. Theorem~\ref{thm:mayer} shows that the radius of convergence of the pressure-activity expansion is approximately $\exp(\beta e_\infty)$ where $e_\infty = \lim (E_k/k)$, as expected from Eq.~\eqref{eq:intro-app}. 
 Theorem~\ref{thm:rhomu} shows that the value $\mu = e_\infty\approx \beta^{-1}\log R^\may(\beta)$  
 corresponds to a cross-over from an exponentially small density to a positive density, thus giving a soft physical 
interpretation to the radius of convergence. 

In the same spirit, Theorems~\ref{thm:cross-over} and~\ref{thm:virconv} give an interpretation of the
 radius of convergence of the pressure-density series. Here there are two possible scenarios: either 
the monatomic gas condenses right away to a solid, or there is an intermediate phase of a polyatomic gas. 
In the latter case the radius of convergence of the virial series is given by the cross-over from monatomic 
to polyatomic gas; in particular, the virial series ceases to converge before the Mayer series does. Props.~\ref{prop:virsin}
 and~\ref{prop:vircon} describe the low-temperature asymptotics of the coefficients of the virial series.

We should stress that the cross-overs corresponding to the radii of convergence of the Mayer and virial series do  not necessarily  
correspond to sharp phase transitions, and may very well be determined by singularities off the positive axis. Nevertheless, they reflect changes in low-temperature physical behavior. This is formally analogous 
to resonances in quantum mechanics, when Green's function singularities off the real axis do not qualify 
as eigenvalues, but can nonetheless affect the system's behavior.

\section{Setting}

We are interested in the statistical mechanics of a classical system of particles, in continuous configuration space, 
interacting via a pair potential $v(|x-y|)$. Thus let $v:[0,\infty) \to \R \cup \{\infty\}$ and 
\begin{equation*}
	U(x_1,\ldots,x_N):= \sum_{1 \leq i <j \leq N} v(|x_i-x_j|)
\end{equation*}
be the energy of a configuration of $N$ particles $x_1,\ldots,x_N\in \R^d$.  We  assume that the energy is \emph{stable}, i.e., there is 
a constant $B\geq 0$ such that 
\begin{equation} \label{eq:stability}
	\forall N\in \N,\ \forall (x_1,\ldots,x_N) \in (\R^d)^N:\ U(x_1,\ldots,x_N) \geq - B N.
\end{equation}
In addition, we assume that the pair potential has finite range, i.e., $v$ has compact support. 
For a given inverse temperature $\beta>0$ and $\Lambda \subset \R^d$, the canonical partition function is
\begin{equation*}
	Z_\Lambda(\beta,N):= \frac{1}{N!} \int_{\Lambda^N} e^{-\beta U(x_1,\ldots,x_N)} \dd x_1\cdots \dd x_N,
\end{equation*}
and the free energy per unit volume is 
\begin{equation*}
	f(\beta,\rho):= -  \lim \frac{1}{\beta |\Lambda|} \log Z_\Lambda(\beta,N). 
\end{equation*} 
The limit is along $N\to \infty$, $\Lambda = [0,L]^d$ with $L\to \infty$, $N/L^d \to \rho >0$. It is well-known that if the potential has no hard core ($r_\mathrm{hc} =0$), 
the limit exists and is finite for all $\rho>0$; if the potential has a hard core, 
then for  a suitable number $\rho_\mathrm{cp}>0$ (the \emph{close-packing density}), the 
limit is finite for $\rho<\rho_\mathrm{cp}$, and infinite for $\rho >\rho_\mathrm{cp}$. 
Moreover the free energy $f(\beta,\rho)$ is a convex function of the density $\rho$. 

The pressure at inverse temperature $\beta$ and chemical potential $\mu \in \R$ is
\begin{equation} \label{eq:legendre}
	p(\beta,\mu) := \sup_{0<\rho<\rho_\mathrm{cp}} \bigl(\rho \mu - f(\beta,\rho) \bigr).  
\end{equation}
We call $\rho(\beta,\mu)$ the maximizer in the previous relation, if it is unique. Because of convexity, the density $\rho(\beta,\mu)$ is an increasing function of the chemical potential $\mu$.

At fixed temperature, for sufficiently negative chemical potential, the pressure is an analytic 
function of the chemical potential and the activity $z$, with expansion 
\begin{equation} \label{eq:mayer}
	\beta p(\beta,\mu) = z + \sum_{n\geq 2} b_n(\beta) z^n,\qquad z = \exp(\beta \mu),
\end{equation}
the \emph{Mayer series}, and the density is given by 
\begin{equation*}
	\rho(\beta,\mu) = z + \sum_{n =2}^\infty n b_n(\beta) z^n. 
\end{equation*} 
Similarly, at low density, the free energy is strictly convex and analytic with expansion
\begin{equation} \label{eq:virial-fe}
	\beta f(\beta,\rho) = \rho (\log\rho - 1) -  \sum_{n \geq 2} d_n(\beta) \rho^n.
\end{equation}
Eq.~\eqref{eq:legendre} gives, for $\mu$ negative enough, 
\begin{equation} \label{eq:virial}
	\beta p(\beta,\mu)= \rho- \sum_{n\geq 2} (n-1) d_n(\beta) \rho^n, \quad \rho = \rho(\beta,\mu),  
\end{equation}
the \emph{virial series}. We would like to know how large $z$, or $\rho$, can be in those equations, and define 
\begin{align*}
	R^\may(\beta):= \sup\bigl\lbrace z>0 \mid \text{Eq.~\eqref{eq:mayer} is true 
		with absolutely convergent series} \bigr\rbrace,\\
	R^\vir(\beta):= \sup\bigl\lbrace \rho>0 \mid \text{Eq.~\eqref{eq:virial-fe} is true 
		with absolutely convergent series} \bigr\rbrace. 
\end{align*}
In principle, $\rmay(\beta)$ and $\rvir(\beta)$ can be smaller than the radius of convergence 
of the corresponding series: Eqs.~\eqref{eq:mayer} and~\eqref{eq:virial-fe} might cease to hold before the series start to diverge. We do not know of any concrete example in our setting, but  there is a simple type of situation where a similar phenomenon arises in mean-field or Landau theories: suppose that a free energy $F(m)$ is the convex envelope of some double-well potential, e.g., $W(m) = (m^2-1)^2$. Then $F(m) =0$ in $|m|<1$, but at $|m|>1$ it becomes equal to $W(m)>0$; in particular, $F(m)$ ceases to equal its (trivial) Taylor expansion around $0$ \emph{before} this expansion ceases to converge. 
 
For non-negative potentials, however, it is known~\cite{leb-pen,penrose} that the domains of convergence coincide with 
the domain of equality of Eqs.~\eqref{eq:mayer},~\eqref{eq:virial-fe}, and~\eqref{eq:virial}, so that in this case
  $\rmay(\beta)$ and $\rvir(\beta)$ are exactly equal to the radius of convergence.

Furthermore we define 
\begin{align}
	\mu_\mathrm{sat}(\beta)&:= \sup\{ \tilde \mu \in \R \mid p(\beta,\mu)\ \text{is analytic in }\mu<\tilde \mu \}  ,\label{eq:musat} \\
	\rho_\mathrm{sat}(\beta)&:= \sup\{ R \in (0,\rho_\mathrm{cp}) \mid f(\beta,\rho)\ \text{is analytic in } 0<\rho<R \}, \label{eq:rhosat}
\end{align}
the chemical potential and density at the onset of condensation, i.e., the quantities associated with saturated gas. In the absence of a phase transition -- for example, in one dimension --, $\mu_\sat(\beta) = \rho_\sat (\beta)= \infty$. 
Another quantity of interest is 
\begin{equation} \label{eq:rhomay}
  \rho^\may(\beta):= \sup\, \{ \rho(\beta,\mu) \mid \exp(\beta \mu)< \rmay(\beta) \}. 
\end{equation}
When there is no phase transition at $\mu^\may(\beta) = \beta^{-1}\log R^\may(\beta)$, then $\rho^\may(\beta) =\rho\bigl(\beta, \mu^\may(\beta) \bigr)$. When there is a phase transition,  the density may have a jump discontinuity and $\rho(\beta,\mu^\may(\beta))$ is no longer well-defined; Eq.~\eqref{eq:rhomay} states that in this case $\rho^\may(\beta)$ equals the left limit of $\rho(\beta,\mu)$ at $\mu^\may(\beta)$. 

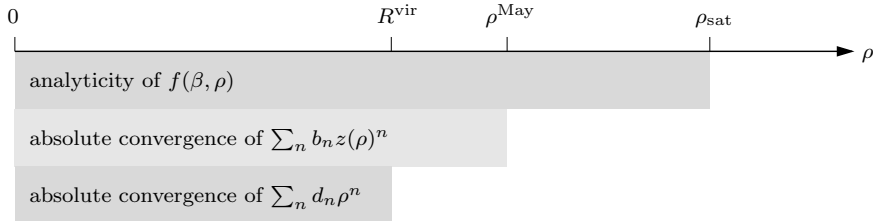
\begin{figure}[htb] \label{figure}
\centering
\begin{picture}(0,0)%
\includegraphics{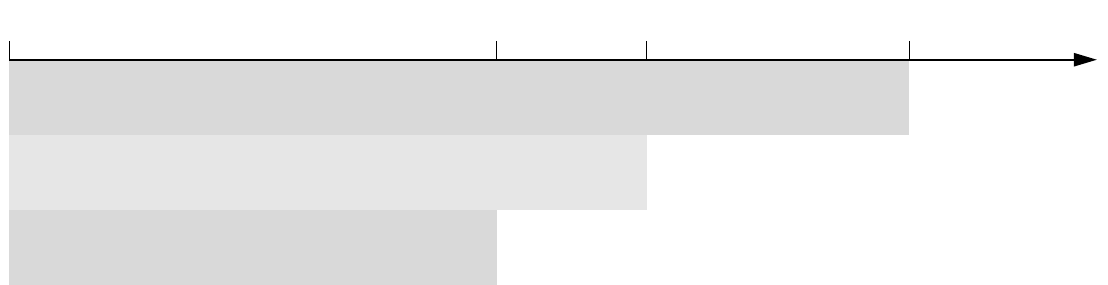}%
\end{picture}%
\setlength{\unitlength}{2368sp}%
\begingroup\makeatletter\ifx\SetFigFont\undefined%
\gdef\SetFigFont#1#2#3#4#5{%
 \reset@font\fontsize{#1}{#2pt}%
 \fontfamily{#3}\fontseries{#4}\fontshape{#5}%
 \selectfont}%
\fi\endgroup%
\begin{picture}(8850,2269)(1126,-4562)
\put(1126,-2461){\makebox(0,0)[lb]{\smash{{\SetFigFont{8}{9.6}{\rmdefault}{\mddefault}{\updefault}{\color[rgb]{0,0,0}0}%
}}}}
\put(8251,-2461){\makebox(0,0)[lb]{\smash{{\SetFigFont{8}{9.6}{\rmdefault}{\mddefault}{\updefault}{\color[rgb]{0,0,0}$\rho_{\rm sat}$}%
}}}}
\put(1351,-4336){\makebox(0,0)[lb]{\smash{{\SetFigFont{8}{9.6}{\rmdefault}{\mddefault}{\updefault}{\color[rgb]{0,0,0}absolute convergence of $\sum_n d_n \rho^n$}%
}}}}
\put(1351,-3736){\makebox(0,0)[lb]{\smash{{\SetFigFont{8}{9.6}{\rmdefault}{\mddefault}{\updefault}{\color[rgb]{0,0,0}absolute convergence of $\sum_n b_n z(\rho)^n$}%
}}}}
\put(1351,-3136){\makebox(0,0)[lb]{\smash{{\SetFigFont{8}{9.6}{\rmdefault}{\mddefault}{\updefault}{\color[rgb]{0,0,0}analyticity of $f(\beta,\rho)$}%
}}}}
\put(9976,-2836){\makebox(0,0)[lb]{\smash{{\SetFigFont{8}{9.6}{\rmdefault}{\mddefault}{\updefault}{\color[rgb]{0,0,0}$\rho$}%
}}}}
\put(4951,-2461){\makebox(0,0)[lb]{\smash{{\SetFigFont{8}{9.6}{\rmdefault}{\mddefault}{\updefault}{\color[rgb]{0,0,0}$R^{\rm vir}$}%
}}}}
\put(6076,-2461){\makebox(0,0)[lb]{\smash{{\SetFigFont{8}{9.6}{\rmdefault}{\mddefault}{\updefault}{\color[rgb]{0,0,0}$\rho^{\rm May}$}%
}}}}
\end{picture}%
\caption{\small Density axis, radii of convergence of the virial expansion ($R^{\rm vir}$) and of the cluster expansion ($\rho^{\rm May}$), and domain of analyticity of the free energy.}
\end{figure}

We have the general bounds, illustrated in Figure~\ref{figure},
\begin{equation*}
	\beta^{-1} \log \rmay(\beta) \leq \mu_\mathrm{sat}(\beta),\qquad 
	\max\Bigl(\rvir(\beta), \rho^\may(\beta)\Bigr) \leq \rho_\mathrm{sat}(\beta). 
\end{equation*}
With these notations, we can ask:
\begin{quote}
	\textbf{Question:} When are the previous  inequalities strict? 
	When they are strict, is it nevertheless possible to give physical meaning to $\rmay$ and $\rvir$, 
even though in this case $\rmay$ and $\rvir$ do not correspond to phase transitions (i.e., points of non-analyticity) ? 
\end{quote}
The main goal of this article is to show that the answer to the second question should be yes; moreover, in the presence of a phase transition for attractive potentials, the inequalites should be approximate equalities, in a sense specified in the Corollary~\ref{cor:where} and the Conjectures~\ref{conj:mayer} 
and~\ref{conj:virial} below. \\

We conclude this section with a description of the convergence criterion for the Mayer series that we shall use. 
Let  
\begin{equation} \label{eq:groundstate}
  E_N:= \inf_{x_1,\ldots,x_N\in (\R^d)^N} U(x_1,\ldots,x_N),\quad E_1 = 0,
\end{equation} 
be the ground state energy for $N$ particles (without any volume constraint), and 
\begin{equation}\label{eq:einf}
  e_\infty:= \inf_{N\in \N} \frac{E_N}{N} = \lim_{N\to \infty} \frac{E_N}{N} \leq 0 
\end{equation}
(note that $(E_N)$ is subadditive). 
The stability assumption on the pair potential ensures that $e_\infty >-\infty$, and Eq.~\eqref{eq:stability}
holds with $B =- e_\infty$ as optimal constant. 
We shall make repeated use of the following Theorem, which is a direct consequence of
~\cite[Theorem 2.1]{pog-ue}, see also~\cite{bf} for integrable potentials (without hard core).

\begin{theorem}[Mayer series estimates \cite{pog-ue}] 
	Let $v(|x-y|)$  be a stable pair interaction potential and $r_\mathrm{hc}\geq 0$ the radius of the hard core. Set 
	\begin{equation*}
		|||v|||:= |B(0,r_\mathrm{hc})| + \int_{\R^d\backslash B(0,r_\mathrm{hc})}\bigl|v( |x|)\bigr| \dd x
	\end{equation*}
	where $B(0,r_\mathrm{hc})$ is the ball of radius $r_\mathrm{hc}$ centered at $0$. Then 
	\begin{equation} \label{eq:maycrit} 
		R^\may(\beta) \geq \frac{e^{\beta e_\infty}}{\beta e |||v|||} 
	\end{equation}
	and for every $0 \leq z \leq \exp(\beta e_\infty)/(\beta e |||v|||)$,
	\begin{equation} \label{eq:remainder}
		\sum_{n\geq 2} n |b_n(\beta)| z ^{n-1} \leq  (e - 1)  e^{- \beta e_\infty}. 
	\end{equation}
	\end{theorem}
As an immediate consequence, we note that 
\begin{equation} \label{eq:musat-lb}
	\liminf_{\beta \to \infty} \mu_\sat (\beta) \geq \liminf_{\beta \to \infty} \beta^{-1} \log \rmay(\beta)  \geq e_\infty.
\end{equation}
Therefore every chemical potential $\mu<e_\infty$, as $\beta \to \infty$, eventually falls into the gas phase. 

\section{Results and conjectures} \label{sec:results}

Our standard assumptions on the potential are the following: 

\begin{assumption}[Minimal assumptions] \label{ass:pairpot} 
	$v:[0,\infty)\to \R\cup\{\infty\}$ satisfies the following assumptions: 
	\begin{itemize}
		\item The energy is stable in the sense of Eq.~\eqref{eq:stability}.	
		\item $v$ is everywhere finite except possibly for a hard core: there is a $r_\mathrm{hc}\geq 0$ such that $v(r)=\infty$ for $r<r_\mathrm{hc}$ and $v(r)<\infty$ 
		for $r>r_\mathrm{hc}$. 
		\item $v$ has compact support, $b:= \sup \{r>0\mid v(r) \neq 0\}<\infty$. 
		\item $v$ is  continuous in $[r_\mathrm{hc},\infty)$, i.e., either 
			$v(r_\mathrm{hc}) = \infty$ and $v(r) \to \infty$ as $r\searrow r_\mathrm{hc}$, 
			or $v(r_\mathrm{hc}) <\infty$ and $v(r) \to v(r_\mathrm{hc})$ as 
			$r\searrow r_\mathrm{hc}$. 			
		\item $v$ has an attractive tail: for suitable $\delta>0$ and all $r \in (b-\delta,b)$, 
			$v(r)<0$. 
	\end{itemize}
\end{assumption}
 Note that we allow for $v(r_\mathrm{hc})<\infty$, which is relevant for  Radin's soft disk potential~\cite{radin}. 
The continuity of the potential is assumed in order to simplify statements on the low-temperature 
asymptotics. In particular, for sufficiently large volumes,
\begin{equation*}
	\lim_{\beta \to \infty} \frac{1}{\beta} \log Z_{\Lambda}(\beta,N) = - E_N.
\end{equation*}
A similar statement holds of course without continuity assumptions, provided that 
we replace the infimum in the definition~\eqref{eq:groundstate} by an essential infimum; 
the continuity allows us to avoid this measure-theoretic complication. \\

Assumption~\ref{ass:pairpot} will be enough when working in the low density gas phase. For results that hold all the way up into a finite density region, we will make additional assumptions. 
We refer to every minimizer $(x_1,\ldots,x_N) \in (\R^d)^N$ of $U(x_1,\ldots,x_N)$ as an $N$-particle \emph{ground state}. Note that the attractive tail favors configurations where particles stick together. 

\begin{assumption}[Ground state geometry and H{\"o}lder continuity] \label{ass:uniform}
	For suitable $a>0$, $r_0 >0$, and every $N\in \N$, there is an
	$N$-particle ground state $(x_1,\ldots,x_N) \in (\R^d)^N$ such that 
	\begin{itemize}
		\item the interparticle distance is bounded below by $r_0$:  
			for all $i \neq j$, $|x_i-x_j| \geq r_0$;
		\item the ground state fits into a cube of volume $N a^d$:
			 $x_1,\ldots,x_N \in [0,N^{1/d} a]^d$.
	\end{itemize}
	Moreover $v(r)$ is uniformly H{\"o}lder continuous in $[r_0,\infty)$. 
\end{assumption}
The simplest example, in dimension two, of a potential satisfying Assumptions~\ref{ass:pairpot} and~\ref{ass:uniform} is Radin's soft disk potential~\cite{radin}, which involves a hard core and a suitable attractive part.  More general potential classes, again in dimension two, are given in~\cite{theil}. 

Assumption~\ref{ass:uniform} is  enough to ensure that various limits 
 can be interchanged. In particular, if 
\begin{equation*}
	e(\rho):= \lim_{\beta \to \infty} f(\beta,\rho)
\end{equation*}
is the ground state energy per unit volume at density $\rho$, then 
\begin{equation*}
	e_\infty = \min_{0<\rho< \rho_\mathrm{cp}} \frac{e(\rho)}{\rho}. 
\end{equation*}
Moreover $e(\rho)/\rho$ has a minimizer $\rho^*\leq 1/a^d$, i.e., the ground state has a finite preferred density. In~\cite{radin,theil}, $\rho^*$ is the density of particles in a simple hexagonal lattice.  \\

Our first result is about the low-temperature behavior of the Mayer coefficients 
and should be contrasted with the alternating sign property for 
non-negative potentials~\cite[Chapter 4]{ruelle-book}. Recall the ground state energies $E_N$ from Eq.~\eqref{eq:groundstate}. 
\begin{theorem}[Mayer coefficients at low temperature] \label{thm:maylt}
	Suppose that $v$ satisfies Assumption~\ref{ass:pairpot}.  
	Then, for every fixed $k$, as $\beta \to \infty$,  $b_k(\beta)$ is eventually positive, 
	and 
	\begin{equation}
		\lim_{\beta \to \infty} \beta^{-1} \log b_k(\beta) = - E_k.
	\end{equation}
\end{theorem}
Thus we may think of the Mayer series as 
\begin{equation} \label{eq:approx}
	\beta p \approx \sum_k z^k \exp(- \beta E_k).
\end{equation}
The subsequent results
 are best understood with the approximate formula~\eqref{eq:approx} in mind. 
We should stress that the approximation~\eqref{eq:approx} can be derived without using Mayer expansions, see~\cite{cly} for a quantum Coulomb systems result. Direct proofs are, in fact, much more instructive from a physical point of view; therefore Theorem~\ref{thm:maylt} 
should be seen as a verification of the consistency of the Mayer series with the approximation~\eqref{eq:approx}.

The next theorem builds upon a low temperature, low density result from~\cite{jkm}
which we briefly recall. Suppose that $v$ satisfies Assumptions~\ref{ass:pairpot} and~\ref{ass:uniform}. Then, for suitable $\beta_0,\rho_0, C_0>0$ and all  
 $\beta \geq \beta_0$ and  $\rho<\rho_0$, 
\begin{equation} \label{eq:free-energy}
	\left| f(\beta,\rho) - \rho \inf_{k \in \N} \frac{E_k+ \beta^{-1} \log \rho}{k} \right|
		\leq C_0 \rho \beta^{-1} \log \beta. 
\end{equation}
As explained in Appendix~\ref{app:free-energy}, $\rho_0$ should be thought of as 
the preferred ground state density (an upper bound is $\rho_0 \leq 1/a^d$ with $a$ 
as in Assumption~\ref{ass:uniform}). The inverse temperature $\beta_0$ is essentially 
determined by the condition $\exp(-\beta \nu^*) \leq 1/(a+R)^d$, where 
%
\begin{equation} \label{eq:nustar}
	\nu^*:= \inf_{k \in \N} (E_k - k e_\infty) \geq 0.
\end{equation}
For potentials with an attractive tail, we have $\nu^*>0$ \cite{jkm}. 

\begin{theorem}[Density increase around $\mu =e_\infty$] \label{thm:rhomu} 
	Suppose that $v$ satisfies Assumptions~\ref{ass:pairpot} and~\ref{ass:uniform}, 
	and that $x\mapsto v(|x|)$ is integrable in $|x|>r_\mathrm{hc}$.
	Let $C_0, \beta_0,\rho_0>0$ be such that Eq.~\eqref{eq:free-energy} 
	holds for all $\beta \geq \beta_0$ and $\rho<\rho_0$. 
	Then:
	\begin{itemize}
		\item For every $C>C_0$ and suitable $\beta_C \geq \beta_0>0$:  
			\begin{equation*}
			\forall \beta \geq \beta_C\quad  \forall 	\mu \geq e_\infty +  C \beta^{-1} \log \beta:\quad
					\rho(\beta,\mu) \geq \frac{C - C_0}{C + C_0}\,  \rho_0.		
			\end{equation*}
		\item For every $C>1$, all $n\in \N$ and suitable $\beta(n,C)$:
			\begin{equation*}
				\forall \beta \geq \beta(n,C)\quad \forall \mu \leq e_\infty - C \beta^{-1} \log \beta:\quad  \rho(\beta,\mu) \leq \beta^{-n}. 
			\end{equation*}
	\end{itemize}
\end{theorem}

In particular, for every fixed $\mu>e_\infty$, as $\beta \to \infty$, 
the density is bounded away from zero,
while for $\mu<e_\infty$, it vanishes exponentially fast (Eq.~\eqref{eq:rhodec} anticipates on Theorem~\ref{thm:cross-over}): 
\begin{alignat}{2}
	\notag \mu>e_\infty&:&\  \liminf_{\beta \to \infty} \rho(\beta,\mu) &\geq \rho_0 >0. \\
	\label{eq:rhodec} \mu<e_\infty&:& \rho(\beta,\mu) &= O( e^{-\beta \nu^*}).
\end{alignat}

\begin{remark}[Non-negative potentials]
	When $v\geq 0$, a similar change in the density behavior occurs around 
	$\mu =0$, as the following two examples illustrate.  
	For an ideal gas in continuum space, $\beta p =z$, $\rho =z$, $e_\infty=0$. 
	For a lattice gas with no interaction except the hard-core 
	on-site repulsion,
	\begin{equation*}
		\beta p(\beta,\mu) = \log(1+z),\quad  \rho(\beta,\mu) = \frac{z}{1+z},\quad 
		e_\infty = 0. 
	\end{equation*}
	As $\beta \to \infty$, if $\mu >0$ is fixed,  the density diverges (for the ideal gas) 
	or approaches the maximum density (for the lattice gas).
	For both the continuum and lattice gas, at fixed $\mu <0$, 
	the density goes to $0$ exponentially fast, but in contrast with the 
	attractive potential case Eq.~\eqref{eq:rhodec} there is no positive lower bound on the 
	rate of exponential decay, $\nu^*=0$. 
\end{remark}

A first consequence is an indication where the low temperature /  low density solid-gas transition is located, if such a phase transition takes place. 
\begin{corollary}[Where to look for a solid-gas transition] \label{cor:where}
	Under the assumptions of Theorem~\ref{thm:rhomu},
	if $\rho_\sat(\beta) \to 0$ as $\beta \to \infty$, then  
	\begin{equation*}
		\mu_\sat(\beta) = e_\infty + O(\beta^{-1} \log \beta) 
	\end{equation*}
	as $\beta \to \infty$. 
\end{corollary}
The lower bound  in the corollary follows from Eq.~\eqref{eq:maycrit}, noting that $\mu_\sat(\beta)\geq\beta^{-1}\log R^\may(\beta)$. For the upper bound, suppose by contradiction that $\mu_\sat(\beta) - e_\infty \gg \beta^{-1}\log \beta$ as $\beta \to \infty$. Then Theorem~\ref{thm:rhomu} tells us that $\rho_\sat(\beta)$ is bounded from below by some positive constant times $\rho_0$, contradicting the assumption $\rho_\sat(\beta) \to 0$. Thus $\mu_\sat(\beta) - e_\infty = O(\beta^{-1} \log \beta)$.

\begin{remark}[Lee-Yang theorem] For a lattice gas on $\Z^d$ with at most 
	one particle per lattice site and 
	 attractive pair interactions $v(x-y) \leq 0$, 
	the Lee-Yang theorem~\cite[Theorem 5.1.3]{ruelle-book} says that if there is a phase transition, then it must be at a chemical potential $\mu$ that satisfies
	\begin{equation*}
		\exp\Bigl(\beta \mu - \frac{1}{2} \sum_{x \neq 0} \beta v(x) \Bigr) =1, 
	\end{equation*}
	i.e.,  $\mu = (\sum_{x \neq 0} v(x) ) /2$. The right-hand side 
	of the latter equality is readily identified with $e_\infty$, the ground state energy per particle for the lattice gas. Thus Corollary~\ref{cor:where} compares well with the Lee-Yang theorem. 
\end{remark}

\begin{remark}[Widom-Rowlinson model] 
	It is instructive to look at a continuum space model for which the existence of a phase transition is known,
      the \emph{Widom-Rowlinson model}~\cite{wr}, see the review~\cite{samaj-wr}.
	Consider particles interacting via the energy 
	\begin{equation*}
		U_\Lambda(x_1,\ldots,x_N) = \bigl| \Lambda 
			\cap \cup_{i=1}^N B(x_i,1)\bigr| - N \bigl| B(0,1)\bigr|
	\end{equation*}
	wher $B(x,1)$ is the ball of radius $1$ centered at $x$. 
	The interaction is not a sum of pair interactions, but it qualifies nevertheless as an attractive, stable,  finite-range interaction. The ground state energy per particle is $e_\infty = - |B(0,1)|$.
	An equivalent formulation is in terms of a two-species model with hard core repulsion between particles of different type: 
	\begin{multline*}
		\sum_{N=0}^\infty\frac{z^N}{N!} \int_{\Lambda^N} e^{- \beta U_{\Lambda}(\vect{x})} \dd \vect{x} \\ 
		= e^{-z_2|\Lambda|} 
		\sum_{N_1,N_2=0}^\infty \frac{z_1^{N_1}}{N_1!} \frac{z_2^{N_2}} {N_2!} 
			\int_{\Lambda^{N_1}}  \int_{\Lambda^{N_2}}  
			\mathbf{1}\bigl(\dist(\vect{x},\vect{y}) \geq 1\bigr)\dd \vect{x}\dd \vect{y}, 
	\end{multline*}
  	provided 
	\begin{equation*}
		\beta = z_2,\quad z =z_1 e^{- z_2 |B(0,1)|}.  
	\end{equation*}
	It is known~\cite{ruelle-wr,cck} that for sufficiently high, equal activities $z_1=z_2$, the system has a phase transition.
    In the one-species picture, a phase transition happens at low temperature and activity $z = \beta \exp( - \beta |B(0,1)|)$, or chemical potential 
	\begin{equation*}
		\mu = - |B(0,1)| + \beta^{-1} \log \beta = e_\infty + \beta^{-1}\log \beta.  
	\end{equation*}
	Again, this matches Corollary~\ref{cor:where}. 
\end{remark} 

A second consequence of Theorem~\ref{thm:rhomu} is that, even when there is no 
phase transition -- for example, in one dimension --, there is nevertheless 
a change in physical behavior around $\mu = e_\infty$: 
consider the family of curves 
$\mu\mapsto \rho(\beta,\mu)$ around $\mu = e_\infty$. 
At $\beta = \infty$, it has a jump of size $\geq \rho_0$. At $\beta$ large but finite, 
there could be either a jump, or the curves
resemble the occupation numbers of fermions around the Fermi energy. 
Hence there is either a phase transition, or a fast increase from small to large density.

We would like to propose this as a a physical interpretation to the domain of convergence 
of the Mayer series, for attractive potentials, based on the following conjecture: 
\begin{conjecture}[Mayer series' radius of convergence] \label{conj:mayer}
	Suppose that $v$ satisfies Assumptions~\ref{ass:pairpot} and~\ref{ass:uniform}. Then
	\begin{equation} \label{eq:conj-mayer}
		 \lim_{\beta \to \infty} \beta^{-1} \log \rmay(\beta) = e_\infty.
	\end{equation}
\end{conjecture} 
Note that for pair potentials whose finite part is integrable, 
we have the lower bound 
Eq.~\eqref{eq:musat-lb} on the liminf.
Hence the only part that is open in the previous conjecture is an upper bound on the limsup.

In fact, a rigorous statement is available for the radius of convergence of the finite volume Mayer series. It is proven by combining Theorem~\ref{thm:maylt} with the bounds from~\cite{penrose}. First recall that the pressure $\beta p_\Lambda(\beta,z)$, 
defined via the logarithm of the grand-canonical partition function in a finite box $\Lambda = [0,L]^d$, has a Mayer expansion similar to Eq.~\eqref{eq:mayer}, with volume-dependent radius of convergence $\rmay_\Lambda(\beta)$. Note that 
 $\liminf_{|\Lambda| \to \infty} \rmay_\Lambda(\beta) \leq \rmay(\beta)$, 
 see~\cite[Eq. (4.2)]{penrose}. 

\begin{theorem}\label{thm:mayer}
	Let the pair interaction satisfy Assumptions~\ref {ass:pairpot} and~\ref{ass:uniform}. 
	Then, if we let first $|\Lambda|\to \infty$ along cubes, and then $\beta \to \infty$, 
	\begin{equation*}
		\lim_{\beta \to \infty} \limsup_{|\Lambda|\to \infty}  \beta^{-1} 
				\log \rmay_\Lambda(\beta) 
		= \lim_{\beta \to \infty} \liminf_{|\Lambda|\to \infty}  \beta^{-1} 
				\log \rmay_\Lambda(\beta) =  e_\infty.
	\end{equation*}
\end{theorem}

We are now heading towards similar interpretations for the virial expansion. 
First, however, we need to understand better the gas phase $\mu <e_\infty$. The reason is 
that inside the gas phase, there might be ``chemical'' transitions~\cite{hill2}, for example, 
from monatomic to diatomic gas. When this happens, the radius of convergence of the virial series 
is determined by that cross-over, and the virial series ceases to converge 
\emph{before} any sharp phase transition is observed, see Theorem~\ref{thm:virconv}. 

The next theorem is a grand-canonical version 
of results from~\cite{jkm} and should be compared to the atomic or molecular limit 
for quantum Coulomb systems~\cite{bm,cly,fefferman}. (See also 
 a result for the classical one-dimensional 
two-component plasma~\cite{lenard}.)  Recall the quantity $\nu^*>0$ defined in  Eq.~\eqref{eq:nustar}.

\begin{theorem}[Possible cross-overs inside the gas phase] \label{thm:cross-over}
	Suppose that $v$ satisfies Assumption~\ref{ass:pairpot} and that $v$ 
	is integrable in $|x|>r_\mathrm{hc}$. Then for every fixed 
	$\mu < e_\infty$,  
	\begin{equation} \label{eq:cross-density}
		\lim_{\beta \to \infty} \beta^{-1} \log \rho(\beta,\mu) = - \inf_{k\geq 1}(E_k - k\mu) < - \nu^* <0.
	\end{equation}
	If in addition $(E_k -k \mu)_{k\in \N}$ has a unique minimizer $k(\mu) \in \N$, then 
	as $\beta \to \infty$, 
	\begin{equation}  \label{eq:cross-pressure}
		\beta p(\beta,\mu) = \frac{\rho(\beta,\mu)}{k(\mu)}\, (1 + o(1)).
	\end{equation}
\end{theorem}

The interpretation is that the gas is, approximately,
 an ideal gas of molecules consisting of $k(\mu)$ particles each, with effective activity 
$z^k \exp(- \beta E_k)$, see also Eq.~\eqref{eq:approx}. An illustration with coexistence curves in the density--temperature plane is given in~\cite[Figure 4]{jkm}. 

\begin{remark}
By now we have two auxiliary variational problems: in Eq.~\eqref{eq:cross-density}, 
we minimize $E_k- k \mu$ with respect to  $k$ at fixed $\mu$; in Eq.~\eqref{eq:free-energy}, we minimize $(E_k + \beta^{-1} \log \rho)/k$ with respect to $k$ at fixed $\beta$ and $\rho$. In Appendix~\ref{app:variational}, we show 
 that these two minimization problems are equivalent and discuss their properties.
\end{remark}

Theorem~\ref{thm:cross-over} covers two different scenarios, depending on the value of 
\begin{equation} \label{eq:mu1}
	\mu_1 := \inf_{k \geq 2} \frac{E_k}{k-1}.
\end{equation}
The quantity $\mu_1$ separates a region dominated by monomers from 
a region dominated by larger groups of particles: for fixed $\mu < \mu_1$, 
$E_k - k \mu$ has the unique minimizer $k(\mu) =1$, and for $\mu>\mu_1$, 
every minimizer is $\geq 2$, see Lemma~\ref{lem:phases}. 
As a consequence, for sufficiently negative chemical potentials, we observe a monatomic gas ($k(\mu) =1$). If $\mu_1=e_\infty$, this is all we see in the gas phase. If $\mu_1<e_\infty$,
as we increase the chemical potential,  
we observe a transition from monatomic to polyatomic gas before the gas condenses. 

The existence of such a transition becomes very natural when we look at a concrete example, taken from~\cite[Sect. 6]{ckms}. \label{two-well}
 Consider a pair potential with a hard core and two potential wells, a deep well at small distances, and a shallow well at larger distances, 
separated by a repulsive ($v>0$) part at intermediate distances. The deep well favors small groups of particles (pairs, triangles or tetraeders, depending on the dimension), arranged at larger distances 
determined by the shallow well. We may think of a solid made of molecules instead of atoms. It is natural, then, that the solid forms after atoms gather in molecules. A rigorous statement with a proof of $\mu_1<e_\infty$, 
 in dimension one, can be found in~\cite{ckms}.

The previous example suggests a relationship between the geometry of ground states  
and the existence of a cross-over inside the gas phase. An interesting open question is, therefore, whether the conditions from~\cite{theil,yfs} ensuring a crystalline ground state with hexagonal lattice (one particle per unit cell) also imply $\mu_1= e_\infty$. A much weaker result is the following: 

\begin{prop}[Sufficient criterion for the absence of polyatomic gas] \label{prop:criterion}
	Let $v$ be a stable pair interaction with attractive tail.  
	\begin{enumerate}
		\item If for all $m,n\in \N$,  
		\begin{equation} \label{eq:gluing}
				E_{m+n+1} \leq E_{m+1}+ E_{n+1}
		\end{equation}
		then $\mu_1 = e_\infty$. 
		\item Suppose that $v(r)$ has a hard core $r_\mathrm{hc}>0$ and $v(r) \leq 0$ for $r\geq r_\mathrm{hc}$.  Then, in dimension $d=1$, the inequality~\eqref{eq:gluing} 
		holds for all $m,n$, and we have $\mu_1 = e_\infty$. 
	\end{enumerate}
\end{prop}
Eq.~\eqref{eq:gluing} should be read with a ``gluing'' operation in mind: instead of juxtaposing $m$ and $n$-particle configurations in space, as is done in order to derive the subadditivity 
$E_{m+n} \leq E_m + E_n$, we glue two configurations with $m+1$ and $n+1$ particles in one point. 

\begin{remark} We owe to G. Friesecke the following remark: in statement 2. of Prop.~\ref{prop:criterion}, we may replace the assumption that $v$ has a hard core by a statement of the type ``$v$ is sufficiently repulsive near the origin'', formulated for example through inequalities for derivatives of the potential. Precise statements are, already in dimension $1$, surprisingly involved. 
\end{remark}

After this excursion into ground states, let us come back to the virial expansion 
and compare the density of saturated gas $\rho_\mathrm{sat}$ of Eq.~\eqref{eq:rhosat} 
with the virial radius of convergence $\rvir$ and with $\rho^\may$ defined in Eq.~\eqref{eq:rhomay}. 
Recall $\nu^*>0$ from Eq.~\eqref{eq:nustar} and let $\nu_1 := - \mu_1$ with $\mu_1$ 
as in Eq.~\eqref{eq:mu1}. The quantity $\nu_1$ is a canonical version of the grand-canonical threshold $\mu_1$. We note that in general $\nu_1 \geq \nu^*$, and $\nu_1>\nu^*$ if and only if $\mu_1 <e_\infty$, 
i.e., if and only if there is a monatomic-polyatomic cross-over inside the gas phase
(see Lemma~\ref{lem:thresholds}). 

\begin{theorem}[Comparison of $\rvir,\rho^\may,\rho_\sat$ and $\nu_1$]
\label{thm:virconv}
	Suppose that $v$ satisfies Assumptions~\ref{ass:pairpot} and~\ref{ass:uniform}. 
	Then
	\begin{align}
		\liminf_{\beta \to \infty} \beta^{-1} \log \rho_\sat(\beta) 
		& \geq \liminf_{\beta \to \infty} \beta^{-1} \log \rho^\may(\beta) 
		\geq - \nu^*, \label{eq:vir-a}\\
	 	 \liminf_{\beta \to \infty} \beta^{-1} \log \rvir(\beta) &\geq - \nu_1. \label{eq:vir-b}
	\end{align} 
	If in addition $\mu_1 <e_\infty$ and $E_k/(k-1)$ has a unique minimizer, 
	then  
	\begin{equation} \label{eq:vir-c}
		\lim_{\beta \to \infty} \beta^{-1} \log \rvir(\beta) = - \nu_1<-\nu^*,
	\end{equation}
	and as $\beta \to \infty$, $R^\mathrm{vir}(\beta) \ll \rho^\may(\beta) \leq  \rho_\sat(\beta)$.
\end{theorem}

Eq.~\eqref{eq:vir-c} tells us that if there is a monatomic-polyatomic cross-over inside the gas phase, then 
that cross-over determines the radius of convergence of the virial expansion, and the virial expansion ceases to converge before there is any phase transition. 

This result should hold without the technical assumption that $E_k/(k-1)$ has a unique 
minimizer. In fact, it is natural to think that Eq.~\eqref{eq:vir-c} extends to the case $\nu_1 = \nu^*$, 
so that the radius of convergence of the virial expansion, 
for attractive potentials, is always determined by the first cross-over -- either from monatomic to polyatomic gas, or directly from small density, monatomic gas, to large density; the latter cross-over possibly being a phase transition (in $d\geq 2$).  
We also have a conjecture on the behavior of  $\rho^\may$ and $\rho_\mathrm{sat}$ 
analogous to Conjecture~\ref{conj:mayer} and Corollary~\ref{cor:where}.

\begin{conjecture} \label{conj:virial}
	For interactions with an attractive tail, 
	\begin{equation*}
		\lim_{\beta \to \infty} \beta^{-1}\log \rho^\may(\beta) = - \nu^*, \quad 
		\lim_{\beta \to \infty} \beta^{-1} \log \rvir(\beta) = - \nu_1 \leq - \nu^*.
	\end{equation*}
	If in addition there is a low-density, low-temperature phase transition, i.e., if $\rho_\sat(\beta) \to 0$ as $\beta \to \infty$, then 
	\begin{equation*}
		\lim_{\beta \to \infty} \beta^{-1} \log \rho_\sat(\beta) = - \nu^*.
	\end{equation*}
\end{conjecture}

Let us recall that the line $\rho = \exp(-\beta \nu^*)$ has the following physical interpretation, proven in~\cite{jkm}: at densities that are very small but higher than $\exp(-\beta \nu^*)$, particles tend to gather in very large clusters (i.e., groups of particles close in space), even though the system is dilute. At densities smaller than $\exp(-\beta \nu^*)$, particles 
stay for themselves or form small groups -- this is the gas phase discussed above. 
%

Finally, we have partial results on the low-temperature asymptotics of the virial coefficients
$d_k(\beta)$ from Eq.~\eqref{eq:virial-fe}, to be compared with  Theorem~\ref{thm:maylt}. 

\begin{prop}[Virial coefficients in the absence of polyatomic gas] \label{prop:virsin}
	Let $v$ satisfy Assumption~\ref{ass:pairpot}. 
	Suppose that Eq.~\eqref{eq:gluing} holds for all $m,n\in\N$. Then 
	$\mu_1 = e_\infty$ and for all $k \geq 2$,  
	\begin{equation*}
		\limsup_{\beta \to \infty} \beta^{-1} \log d_k(\beta) \leq - E_k.
	\end{equation*}
	If in addtion the inequality~\eqref{eq:gluing} is strict for all $m,n\in \N$, the previous inequality for the limsup becomes an equality for the limit. 
\end{prop}

\begin{prop}[Virial coefficients in the presence of a monatomic-diatomic transition] \label{prop:vircon}
	Suppose that $v$ satisfies Assumption~\ref{ass:pairpot}, 
	and in addition $\mu_1<e_\infty$ and $E_k/(k-1)$ has the unique minimizer $k=2$. Then 
	for every $k \geq 2$, as $\beta \to \infty$, $d_k(\beta)$ eventually has the sign 
	of $(-1)^{k}$,  and
	\begin{equation*}
		\lim_{\beta \to \infty} 
                          \beta^{-1} \log\Bigl( (-1)^{k} d_k(\beta) \Bigr)= - (k-1) E_2 > - E_k. 
	\end{equation*}
	In particular, $|d_k(\beta)| / b_k(\beta) \to \infty$ as $\beta \to \infty$.  	
\end{prop}

For the two-well example from p.~\pageref{two-well}, we expect $\mu_1<e_\infty$ 
and $E_k/(k-1)$ should have the unique minimizer $p=d+1$, with $d$ the dimension of the configuration space. 
In one dimension, this is proven~\cite[Sect. 4]{ckms}, and gives an example to which Prop.~\ref{prop:vircon} applies. For higher dimensions, we note that the natural generalization 
of Prop.~\ref{prop:vircon} when $\mu_1 = E_p/(p-1)$ for a unique $p \geq 3$ 
 is 
\begin{equation*}
	d_{1+r(p-1)+q}(\beta) =  (-1)^{r} d_{q+1}(\beta) \exp\bigl(- \beta r (E_p +o(1)\bigr),	
\end{equation*}
$r\in \N_0, \ q= 0,1,\ldots,p-2$. We leave the proof, or disproof,  as an open problem, and do not exclude surprises -- it is not impossible that additional conditions, in the spirit of  Eq.~\eqref{eq:gluing}, are needed. 

In any case, we see that the behavior of the virial coefficients is more complex than that 
of the Mayer coefficients. In the absence of a cross-over, at low temperature, each virial coefficients is eventually positive -- as envisioned by Mayer and Mayer~\cite[Chapter 14(f)]{mayerbook}. Note, however, that in~\cite{mayerbook} the authors argue that there is a temperature below which \emph{all} virial coefficients become positive; this statement is much stronger than Prop.~\ref{prop:virsin}. 
When there is a cross-over, in contrast, the sign of the coefficients varies in a periodic way. 

\section{Mayer coefficients at low temperature} \label{sec:maycoeff}

Here we prove Theorem~\ref{thm:maylt}. 
We use the usual short-hand $v_{ij} = v(|x_i-x_j|)$, and $f_{ij}$ as in 
$$ \exp(- \beta v(|x_i - x_j|)) = \exp(- \beta v_{ij}) = 1 + f_{ij}.$$ 
We recall the expression of the Mayer coefficient: it is known that 
\begin{equation} \label{eq:bk} 
	b_k(\beta)= \frac{1}{k!} \sum_{\gamma\ \text{conn.}} \int_{(\R^d)^{k-1}}  
	\prod_{(ij) \in \gamma}  f_{ij}(\vect{x}) \dd x_2\cdots \dd x_k,\quad x_1:=0.
\end{equation}	
The sum is over connected, undirected graphs  $\gamma = (V,E)$ with vertices $1,\ldots,k$, and $\prod_{(ij) \in \gamma}$ is the product over edges $\{i,j\} \in E$, $i<j$ 
(no self-edges $(ii)$). 

Let us start with a look at the $\beta \to \infty$ behavior for an individual graph.  
Observing that 
\begin{equation*}
	f_{ij}(\vect{x}) = \begin{cases}
				(1+o(1)) \exp( - \beta v_{ij}(\vect{x})),&\quad v_{ij}(\vect{x})<0, \\
				- 1 + o(1),& \quad v_{ij}(\vect{x})> 0,
			\end{cases}
\end{equation*}
we get
\begin{equation*}
	\prod_{(ij)\in \gamma} \bigl| f_{ij}(\vect{x}) \bigr| 
		= (1+o(1)) \exp\Bigl( - \beta \sum_{(ij) \in \gamma} v_{ij}(\vect{x}) 
			\mathbf{1}(v_{ij}(\vect{x})<0) \Bigr). 
\end{equation*}
In the exponent, only negative interactions appear. As a result, we may end up with 
energies much smaller than the ground state energy, seemingly contradicting Theorem~\ref{thm:maylt}. The reason is, of course, that there are cancellations between different graphs. In order to get a hold on them, it is convenient to do separate book-keepings for 
``positive'' and ``negative'' edges. Given $\vect{x} = (x_1,\ldots,x_k)$, we define 
\begin{align*}
	\mathcal{E}^+(\vect{x}) & :=\bigl\lbrace \{i,j\} \mid 1\leq i<j\leq k,\ 
		v_{ij}(\vect{x}) >0   \bigr \rbrace \\
	\mathcal{E}^-(\vect{x}) & :=\bigl\lbrace \{i,j\} \mid 1\leq i<j\leq k,\ 
		v_{ij}(\vect{x}) < 0  \bigr \rbrace. 
\end{align*}
and let $\gamma^\pm(\vect{x})$ be the graphs with 
vertices $1,\ldots,k$ and edge sets  $\mathcal{E}^\pm(\vect{x})$.

The next simplifying observation is that if the interaction has a finite range $R>0$, 
$f_{ij}(\vect{x})$ vanishes as soon as $|x_i-x_j|>R$. Therefore we define, for $\vect{x} = (x_1,\ldots,x_k)$, 
\begin{equation*}
	\mathcal{E}(\vect{x})  :=\bigl\lbrace \{i,j\} \mid 1\leq i<j\leq k,\ 
		|x_i-x_j| \leq R  \bigr \rbrace, 
\end{equation*}
and let $\gamma(\vect{x})$ be the graph with vertices $1,\ldots,k$ and edge set $\mathcal{E}(\vect{x})$. We call a configuration $\vect{x}$ \emph{connected} if the graph $\gamma(\vect{x})$ is connected, and write $\mathbf{1}_\conn(\vect{x})$ for the corresponding characteristic function. 
With these notations, for every configuration $\vect{x}$ and every graph $\gamma$, 
\begin{equation*}
  \prod_{(ij) \in \gamma} f_{ij}(\vect{x}) \neq 0 \ \Rightarrow\ \mathcal{E}(\gamma) \subset \mathcal{E}(\vect{x}), 
\end{equation*}
and if $\gamma$ is connected, so is $\vect{x}$. 

We are going to compare the Mayer coefficient with a partition function for connected configurations, 
\begin{equation} \label{eq:zk}
 	Z_k^\cl(\beta) := \frac{1}{k!}  \int_{(\R^d)^{k-1}} e^{-\beta U(0,x_2,\ldots,x_k)} \mathbf{1}_\text{conn}(0,x_2,\ldots,x_k) \dd x_2\cdots \dd  x_k.
\end{equation}

\begin{lemma}[Cluster partition function vs. Mayer coefficient] \label{lem:zkbk}
	\begin{multline} \label{eq:zkbk}
	 Z_k^\cl(\beta) -  b_k (\beta) \\ 
		=  \frac{1}{k!}  
		 \sum_{\gamma\ \text{not conn.}} \int_{(\R^d)^{k-1}} 
			\prod_{(ij) \in \gamma}
	f_{ij}(\vect{x})\mathbf{1}_\conn(\vect{x}) \dd x_2\cdots 
	\dd x_k,\ x_1 =0
	\end{multline}
	where the sum extends over graphs $\gamma$ with vertices $\{1,\ldots, k\}$ that are \emph{not} connected.
\end{lemma}

\begin{proof}
	In the integral for $Z_k^\cl(\beta)$, write as usual $\exp(-\beta v_{ij})= 1+ f_{ij}$ 
	and expand. This gives a sum over graphs. The graphs that are not connected 
 	correspond to the right-hand side of Eq.~\eqref{eq:zkbk}. The connected graphs 
	yield an integral similar to Eq.~\eqref{eq:bk}, except that there is the additional 
	characteristic function $\mathbf{1}_\conn(\vect{x})$. Noting that 
	$\prod_{(ij)\in \gamma}f_{ij}(\vect{x})$ 
	vanishes if $\gamma$ is connected and $\vect{x}$ is not 
	connected, we can drop the characteristic function without changing the value of 
	the integral, and obtain Eq.~\eqref{eq:zkbk}. 
\end{proof}

%
%

For $\gamma$ a graph with vertex set $\{1,\ldots,k\}$,  and 
$\vect{x}= (x_1,\ldots,x_k) \in (\R^d)^k$ a configuration, write
$\gamma^-(\vect{x}) \cap \gamma$ for the graph with vertices $1,\ldots,k$ 
and edge set $\mathcal{E}(\gamma)\cap \mathcal{E}^-(\vect{x})$. 
Thus $\gamma^-(\vect{x})\cap \gamma$ is the subgraph of $\gamma$ consisting 
of the negative edges. 

\begin{lemma} \label{lem:negsum}
	Let $k\in \N$ and $\gamma^-$ a graph with vertices $1,\ldots,k$ 
	with connected components of size $k_1,\ldots,k_r$, $r\in \N$, $\sum_1^r k_i = k$. 
	Then 
	\begin{equation*}
		\left| \sum_{\gamma:\ \gamma^-(\vect{x})\cap \gamma = \gamma^-} 
		 \prod_{(ij) \in \gamma} f_{ij}(\vect{x}) \right| \leq 
			C_k \exp\bigl( - \beta (E_{k_1}+ \cdots + E_{k_r} ) \bigr).
	\end{equation*}
	for some suitable $C_k>0$ which does not depend on $\beta$ or $r,k_1,\ldots,k_r$. 
	A similar estimate holds, for $r\geq 2$, if the sum is further restricted to graphs $\gamma$ with $\gamma^-(\vect{x})\cap \gamma = \gamma^-$ 
	that are not connected.
\end{lemma}

\begin{remark}
	The lemma is also true for a sum further restricted to graphs that are connected. 
	It becomes wrong, in general, for doubly connected graphs. 
\end{remark}

\begin{proof}
	Consider first the case $r=1$, i.e., $\gamma^-$ connected. Then 
	\begin{align*}
		 \sum_{\gamma:\ \gamma^-(\vect{x})\cap \gamma = \gamma^-}  
			\prod_{(ij) \in \gamma} f_{ij}(\vect{x}) 
		& = \left( \prod_{(ij) \in \gamma^-} f_{ij}(\vect{x})\right) \sum_{\mathcal{E} \subset \mathcal{E}^+(\vect{x})} 
				\prod_{(ij) \in \mathcal{E}} f_{ij}(\vect{x}) \\
		& = \left( \prod_{(ij) \in \gamma^-} f_{ij}(\vect{x})\right)\left( 
	\prod_{(ij) \in \mathcal{E}^ +(\vect{x})}  e^{-\beta v_{ij}(\vect{x})} \right).
	\end{align*}
	Noting that  for a negative edge, $0 \leq f_{ij} \leq \exp(- \beta v_{ij})$, it follows that  
	\begin{align*}
		0 \leq \sum_{\gamma:\ \gamma^-(\vect{x})\cap \gamma = \gamma^-} 
			 \prod_{(ij) \in \gamma} f_{ij}(\vect{x}) 
		 & \leq \exp\Bigl(- \beta \Bigl[ \sum_{(ij) \in \gamma^+(\vect{x})} v_{ij}(\vect{x}) + \sum_{(ij) \in \gamma^-} v_{ij}(\vect{x}) \Bigr] \Bigl) \\
		& \leq \exp\bigl(  - \beta U(x_1,\ldots,x_k)\bigr) \leq \exp(-\beta E_k). 
	\end{align*}
	Next, suppose that $r\geq 2$ and that we restrict the sum to graphs $\gamma$ 
	that have $\gamma^-(\vect{x}) \cap \gamma = \gamma^-$ and 
	are disconnected. Such a graph can be constructed from $\gamma^-$ in two steps: first, 
	add 
	positive intra-component edges, i.e., 
	edges $(ij)$ that have $v_{ij} \geq 0$ and connect two labels $i,j$ belonging 
	to the same connected component of $\gamma^-$. There is no restriction on the choices 
	of such edges. Second, add positive inter-component edges. There is a restriction 
	on how many edges we may add, since the resulting graph is required 
	to be disconnected. Let $\mathcal{F}$ be the collection of allowed inter-component edge sets. 	The sum to be estimated takes the form 
	\begin{equation} \label{eq:fsum}
		\sum_{\mathcal{E} \in \mathcal{F}} \left(\prod_{(ij) \in \mathcal{E}} f_{ij}(\vect{x}) \right) 
		\prod_{q=1}^r \left[\Bigl( \prod_{(ij)\in \gamma_q^-} f_{ij}(\vect{x})\Bigr) \Bigl(\prod_{(ij) \in \gamma_q^+(\vect{x})} e^{-\beta v_{ij}(\vect{x})} \Bigr) \right]
	\end{equation}
	Here $\gamma_1^-,\ldots,\gamma_r^-$ are the connected components of $\gamma^-$, 
	and $\gamma_1^+(\vect{x}),\ldots,\gamma_r^+(\vect{x})$ have as edge set the 
	positive intra-component edges.  Noting that $- 1 \leq f_{ij} \leq 0$ for every 
	positive edge, we deduce that~\eqref{eq:fsum} has absolute value bounded by 
	\begin{equation*}
		|\mathcal{F}|\, \exp \bigl( - \beta (E_{k_1}+ \cdots + E_{k_r}) \bigr).
	\end{equation*}
	Since $|\mathcal{F}|$ can be bounded by some $k$-dependent constant, independent 
	of $\gamma^-$, this concludes the proof.
\end{proof}

\begin{proof}[Proof of Theorem~\ref{thm:maylt}]
	Let $\vect{x} = (x_1,\ldots,x_k)$ be an arbitrary configuration and 
	$\gamma$ a graph with vertices $1,\ldots,k$ that is not connected. Then 
	 $\gamma^-(\vect{x}) \cap \gamma$ is not connected either. Therefore 
	\begin{equation*}
		\sum_{\gamma\ \text{not conn.}} \prod_{(ij) \in \gamma} f_{ij}(\vect{x}) 
		= \sum_{\gamma^-\ \text{not conn.}} \sum_{\stackrel{\gamma\ \text{not conn.}:}
		{\gamma^-(\vect{x}) \cap \gamma = \gamma^-}} 
		\prod_{(ij) \in \gamma} f_{ij}(\vect{x}).
	\end{equation*}
	 Lemma~\ref{lem:negsum} then yields a bound on the absolute value of the form 
	\begin{equation} \label{eq:ckek}
		C_k \exp\bigl(- \beta (E_{k_1} + \cdots + E_{k_r}) \bigr) \leq 
			C_k \exp( - \beta (r-1) \eps) \exp( - \beta E_k).
	\end{equation}
	Here we have used that for potentials with an attractive tail, 
	for suitable $\eps>0$ and all $k,q \in\N$, 
	$E_{k+q} \leq E_k + E_q - \eps$ (see the appendix in~\cite{jkm}). 
	Since the set of connected 
	configurations $(0,x_2,\ldots,x_k)$ has a finite Lebesgue volume, 
	 the integral on the right-hand side of Eq.~\eqref{eq:zkbk} 
	has an upper bound similar to the right-hand side of Eq.~\eqref{eq:ckek}.
	On the other hand, because of the attractive tail of the potential, 
	ground states are always connected. The continuity of the potential therefore yields 
	\begin{equation*}
		\lim_{\beta \to \infty} \beta^{-1} \log Z_k^\cl(\beta) = - E_k, 
	\end{equation*}
	and we conclude from Lemma~\ref{lem:zkbk} that for every fixed $k$, as $\beta \to \infty$, 
	\begin{equation*}
		b_k(\beta) =(1+ O(e^{-\eps \beta})) Z_k^\cl(\beta) = \exp\bigl( - \beta (E_k+o(1))\bigr). \qedhere
	\end{equation*}
\end{proof}

\section{Virial coefficients and absence of polyatomic gas} \label{sec:vircoeff}

The virial series coefficients $d_n(\beta)$ have an expression as  a sum of integrals, indexed by graphs, similar to Eq.~\eqref{eq:bk}; the sole difference is that the sum is over doubly connected graphs only (see, e.g., \cite{hill2}). 
Thus we may write 
\begin{equation*}
	b_n(\beta) = d_n(\beta)+ \text{a sum over graphs that are not doubly connected}.
\end{equation*}
By  Lemma~\ref{lem:zkbk}, the main contribution to 
$b_n(\beta)$ comes from integrals around the ground state. In dimension $2$, the ground state could resemble, for example, a hexagonal lattice and should be doubly connected; thus one 
might think that the main contribution to $b_n(\beta)$ comes from doubly connected graphs 
and therefore $b_n(\beta) = (1+o(1))d_n(\beta)$. 

It turns out that this naive guess leads to the right answer under the assumptions of Prop.~\ref{prop:virsin}, but fails in the setting of Prop.~\ref{prop:vircon}. In the latter case, 
both $d_n(\beta)$ and the sum over graphs that are not doubly connected are much larger, in absolute value, than $b_n(\beta)$. 

The reason for this complex behavior is that the cancellations between different graphs become rather subtle to handle; see also the remark after Lemma~\ref{lem:negsum}.  Therefore we work instead with an expression of the virial coefficients 
in terms of Mayer coefficients. It is known ~\cite[Chapter 5, Eq. (25.30)]{hill1} that
\begin{equation} \label{eq:dnbn}
	d_n(\beta) = \sideset{}{'}\sum_{m_2,\ldots,m_n} (-1)^{(\sum_2^n m_j)-1} 
			a(m_2,\ldots,m_n) b_2(\beta)^{m_2} \cdots b_n(\beta)^{m_n},
\end{equation}
with a sum over $(m_2,\ldots,m_n) \in \N_0^{n-1}$ such that 
\begin{equation} \label{eq:mcond}
	\sum_{j=2}^n (j-1)m_j = n-1,
\end{equation}
and
\begin{equation} \label{eq:aformula}
	a(m_2,\ldots,m_n) = \frac{(n-2 + \sum_2^n m_j)!}{n!} \prod_{j=2}^n 
		\frac{j^{m_j}}{m_j!} >0 . 
\end{equation}
One may check that the contribution of the vector with $m_n =1$ and $m_2 =\cdots = m_{n-1} =  0$  
is equal to $b_n(\beta)$. 

The form of Eq.~\eqref{eq:dnbn}  becomes very natural when the virial series is  derived directly with the help of a cluster expansion in the canonical ensemble, as recently done
in~\cite{pulv-tsag}, see Appendix~\ref{app:clexp}. 

\begin{lemma} \label{lem:eksin}
	Suppose Eq.~\eqref{eq:gluing} holds for all $m,n\in \N$. Then, 
	for all $(m_2,\ldots,m_n) \in \N_0^{n-1}$ satisfying Eq.~\eqref{eq:mcond}, 
	\begin{equation} \label{eq:ekin}
		E_n \leq \sum_{j=2}^n m_j E_j.
	\end{equation} 		
	If the inequality~\eqref{eq:gluing} is strict for all $m,n$, then the previous inequality is 
	strict too. 
\end{lemma}

\begin{proof}
	We proceed by induction over $r= \sum_{2}^n m_j$. For $r=1$, 
	the inequality is a trivial equality. For $r=2$, the inequality~\eqref{eq:ekin}
	is the same as~\eqref{eq:gluing}, and true by assumption. 
	For the induction step, suppose that the statement is true, at $r$, 
	for all $n \in \N$. Let $(m_2,\ldots,m_n)$ satisfy Eq.~\eqref{eq:mcond} 
	and such that $\sum_2^n m_j = r+1\geq 2$. Write 
	$m_j = m'_j+ \delta_{jk}$
	for some $k\in \{2,\ldots,n\}$. Then 
	\begin{align*}
		E_n &= E_{1+ (\sum_2^n (j-1) m'_j) + (k-1)} 
			\leq E_ {1+ \sum_2^n (j-1) m'_j} + E_k\quad 
		 \text{by Eq.~\eqref{eq:gluing}}\\
			&\leq \sum_{j=2}^n m'_j E_j + E_k=   \sum_{j=2}^n m_j E_j, 
	\end{align*}
	which proves the claim. The procedure for strict inequalities is exactly the same. 
\end{proof}	

Prop.~\ref{prop:virsin} is an immediate consequence of Eq.~\eqref{eq:dnbn}, 
Lemma~\ref{lem:eksin} and Theorem~\ref{thm:maylt}. 
\begin{proof}[Proof of Prop.~\eqref{prop:virsin}]
	By Lemma~\ref{lem:eksin} and Theorem~\ref{thm:maylt}, all terms 
	in the sum~\eqref{eq:dnbn} are of order at most $\exp(- \beta E_n( 1+o(1))$. 
	If the inequality~\eqref{eq:gluing} is strict, the dominant contribution comes 
	from $b_n(\beta)$ ($r=1$), which is equal to $\exp(-\beta E_n(1+o(1)))$, again 
	by Theorem~\ref{thm:maylt}. 
\end{proof}

\begin{lemma} \label{lem:vircon}
	Suppose that $\mu_1 = E_2 <e_\infty$ and $E_k/(k-1)>E_2$ for all $k \geq 3$. 
	Then, for all $(m_2,\ldots,m_n) \in \N_0^{n-1}$ satisfying Eq.~\eqref{eq:mcond}, 
	\begin{equation*}
		(n-1) E_2 \leq \sum_{j=2}^n m_j E_j 
	\end{equation*} 
	with equality if and only if $m_2 = n-1$, $m_3=\cdots = m_n =0$.  
\end{lemma}
Thus the main contribution to $b_n(\beta) - d_n(\beta)$ comes from graphs
 whose doubly connected components all have size~$2$. 

\begin{proof}
	By assumption, 
	 $E_k > (k-1) E_2$ for all $k \geq 3$, from which we obtain 
	\begin{equation*}
		E_{k_1} + \cdots + E_{k_r} \geq \sum_1^r (k_i - 1) E_2 = (n-1) E_2	
	\end{equation*}
	for all $k_1,\ldots,k_r \geq 2$, 
	with equality if and only if all $k_i$'s are equal to~2. 
	Writing $m_j$ for the number of $i$'s such that $k_i = j$, we obtain the desired statement. 
\end{proof}

\begin{proof}[Proof of Prop.~\ref{prop:vircon}]
	By Lemma~\ref{lem:vircon} and Theorem~\ref{thm:maylt}, in Eq.~\eqref{eq:dnbn} 
	all terms are negligible except the one for $m_2= n-1$, $m_3=\cdots = m_n =0$. 
\end{proof}

We conclude with the proof of the sufficient criterion for the absence of polyatomic gas. 

\begin{proof}[Proof of Prop.~\ref{prop:criterion}]
1. Eq.~\eqref{eq:gluing} says that $n\mapsto E_{n+1}$ is subadditive, thus 
	\begin{equation*}
		\mu_1 = \inf_{n\in \N} \frac{E_{n+1}}{n} = \lim_{n\to \infty} \frac{E_{n+1}}{n} = e_\infty.
	\end{equation*}
		
	2. Suppose $d=1$ and $v(r) \leq 0$ for all $r \geq r_\mathrm{hc}$ and $v(r) =\infty$ for $r< r_\mathrm{hc}$. 
	 Let $m,n\in \N$ 
	and $\vect{x}$, $\vect{y}$ be $m+1$ and $n+1$-particle ground states.
	Without loss of generality we may assume that $x_1\leq \cdots \leq x_{m+1}$, 
	$y_1\leq \cdots\leq y_{n+1}$, and $y_1 = x_{m+1}$.
	We construct a $m+n+1$-configuration $\vect{z}$ by gluing the two ground states: 
	set $z_1:= x_1$,..., $z_{m+1} := x_{m+1}$ and $z_{m+j}:= y_j$ for $j=1,\ldots,n+1$.  
	Then  
	\begin{equation} \label{eq:doublesum}
		E_{m+n+1} \leq U(\vect{z}) = E_{m+1} + E_{n+1} 
			+ \sum_{j=1}^m \sum_{k=1}^n v(|x_j - y_{k+1}|). 
	\end{equation}
	Because of the hard core, all particles have mutual distance $\geq r_\mathrm{hc}$ and negative interactions.
	Hence the double sum in Eq.~\eqref{eq:doublesum} is smaller or equal to zero, 
	and we deduce $E_{m+n+1} \leq E_{m+1} + E_{n+1}$. 
	 Since $m$ and $n$ were arbitrary, 
	applying the sufficient criterion from 1., we get $\mu_1 = e_\infty$.
\end{proof}

\section{Bounds for the density $\rho(\beta,\mu)$}

In this section we prove Theorems~\ref{thm:rhomu} and~\ref{thm:cross-over}.

\begin{proof}[Proof of Theorem~\ref{thm:rhomu}]
	1.  Suppose $\mu > e_\infty + C\beta^{-1} \log \beta$ with $C>C_0$. 
      Write $\rho = \exp(-\beta \nu)$ and $\mu(\nu) = \inf_{k\in \N} (E_k - \nu)/k$. 
	For $\nu\leq\nu^*$, $\mu(\nu) = e_\infty$, and for $\nu>\nu^*$, $\mu(\nu)<e_\infty$, see Appendix~\ref{app:variational}.  
	Using Eq.~\eqref{eq:free-energy}, 
	\begin{align*}
		p(\beta,\mu) & = \sup_{0<\rho<\rho_\mathrm{cp}} (\mu \rho - f(\beta,\rho)) \\
				& \geq \sup_{0<\rho<\rho_0} (\mu \rho - f(\beta,\rho)) \\
				& \geq \sup_{0<\rho<\rho_0} \Bigl(\mu \rho - \mu(\nu) \rho 
						- C_0 \rho \beta^{-1} \log \beta \Bigr) \\
				& \geq \sup_{0< \rho< \rho_0} 
					\Bigl( ( \mu - e_\infty - C_0\beta^{-1}\log \beta) \rho \Bigr)\\
				& = \rho_0 \bigl( \mu - e_\infty - C_0\beta^{-1}\log \beta\bigr), 
	\end{align*}
	since $\mu>e_\infty + C_0 \beta^{-1} \log \beta$. On the other hand, let 
	$\rho = \rho(\beta,\mu)$ be any maximizer of $\rho \mu - f(\beta,\rho)$. 	
	If $\rho \geq \rho_0$, we are done. If $\exp(-\beta \nu^*) \leq \rho \leq \rho_0$, then 
	\begin{equation*}
		\rho (\mu - e_\infty + C_0 \beta^{-1} \log \beta)  \geq p(\beta,\mu) \geq 
			 \rho_0 \bigl( \mu - e_\infty - C_0\beta^{-1}\log \beta\bigr), 
	\end{equation*}	
	which gives 
	\begin{align*}
		\mu & \leq e_\infty + C_0 \frac{\rho_0 + \rho}{\rho_0 - \rho } \beta^{-1} \log \beta. 
	\end{align*}
	Since $\mu \geq e_\infty + C \beta^{-1} \log \beta$, we obtain 
	$C_0(\rho_0 + \rho) \geq C (\rho_0 - \rho)$ whence 
	$$ \rho \geq \frac{C - C_0 }{C + C_0}\, \rho_0. $$ 
	Thus we are left with the case $\rho<\exp(-\beta \nu^*)$, i.e., $\nu>\nu^*$.
	Noting $\mu(\nu) \geq e_\infty - \nu$ for all $\nu$, we get
	\begin{align*}
		 \rho \mu - \rho (e_\infty - \nu) + C_0 \rho \beta^{-1} \log \beta  
		\geq  p(\beta,\mu) 
		  \geq \rho_0 \bigl( \mu - e_\infty - C_0\beta^{-1}\log \beta\bigr). 
	\end{align*}
	Since $\rho \nu = \nu e^{-\beta \nu} \leq \beta^{-1}$, we obtain
	\begin{align*}
		\mu & \leq e_\infty + \frac{\beta^{-1}+ C_0(\rho_0+\rho) \beta^{-1} \log \beta}{\rho_0-\rho} \\
		 	& \leq e_\infty + \frac{\beta^{-1}+ C_0(\rho_0+e^{-\beta \nu^*}) \beta^{-1} \log \beta}{\rho_0-e^{-\beta \nu^*}} 
			= e_\infty + C_0 (1 +o(1)) \beta^{-1} \log \beta,
	\end{align*}
	which for sufficiently large $\beta$ is 
	in contradiction with the assumption on $\mu$. 

	2. For $\mu<e_\infty - C \beta^{-1} \log \beta$ with $C>1$ we use Theorem~\ref{thm:maylt}. Define 
	$R>0$ by
	\begin{equation} \label{eq:radius}
		R e^{-\beta e_\infty} \beta |||\bar v||| =1/e, 
	\end{equation}
	see Eq.~\eqref{eq:maycrit}. For $K \in \N$, using Eq.~\eqref{eq:remainder}, we have 
	\begin{align*}
		\rho(\beta,\mu) & \leq \sum_{k=1}^K k \exp\bigl( \beta (k \mu - E_k +o (1)) \bigr) 
			+  \left(\frac{z}{R}\right) ^K z \sum_{k=K+1}^\infty k |b_k(\beta)| R^{k-1} \\
			& \leq  e^{- \beta (\nu^* +o(1))} + \left(\frac{z}{R}\right) ^K 
				(e-1) z e^{-\beta e_\infty}\\
			& = e^{- \beta (\nu^* +o(1))}  \\
				& \qquad + 
				(e-1) \exp\Bigl( K\Bigl[ (- C+1)\beta^{-1}\log \beta +\beta^{-1} (1-\log ||| \bar v|||) \Bigr] \Bigr) 
	\end{align*}
	As $\beta \to \infty$, the second term is of order $\beta^{-K (C-1 +o(1))}$. Since 
	$K$ could be chosen arbitrarily large, this completes the proof of Theorem~\ref{thm:rhomu}.
\end{proof}

\begin{proof}[Proof of Theorem~\ref{thm:cross-over}]
	We proceed analogously to the proof of the second part of Theorem~\ref{thm:rhomu}. 
	Let $R = \exp(\beta (e_\infty +o(1)))$ be as in Eq.~\eqref{eq:radius}. 
	Fix $\mu <e_\infty$. 
	For $K \in \N$, using Eq.~\eqref{eq:remainder},
		\begin{equation} \label{eq:co-1}
			\Bigl|\rho(\beta,\mu) - \sum_{k=1}^K k \exp\bigl( \beta (k \mu - E_k +o (1)) \bigr)\Bigr| 
				\leq  \left(\frac{z}{R}\right) ^K 
					(e-1) z e^{-\beta e_\infty}.
		\end{equation}
	Since $E_k - k \mu = k (e_\infty - \mu+o(1))  \to \infty$ 
	as $k \to \infty$, there is a finite $k(\mu)$ minimizing $E_k - k \mu$. 
	Choosing $K$ large enough so that $K \geq k(\mu)$ and $K(e_\infty - \mu) > | \inf_k(E_k - k \mu) |$,
	Eq.~\eqref{eq:cross-density} follows from the inequality~\eqref{eq:co-1}. 
	
	If $(E_k - k \mu)_{k \in \N}$ has a unique minimizer $k(\mu) \in \N$, the previous argument actually yields 
	\begin{equation*}
		\rho(\beta, \mu) =  k(\mu) b_{k(\mu)}(\beta) e^{\beta k(\mu) \mu} (1+o(1)). 
	\end{equation*}
	An analogous argument gives 
	\begin{equation*}
		\beta p(\beta, \mu) =  b_{k(\mu)}(\beta) e^{\beta k(\mu) \mu} (1+o(1)), 
	\end{equation*}
	and Eq.~\eqref{eq:cross-density} follows.  
\end{proof}

\section{Asymptotics of $\rmay$, $\rho_\sat$, $\rho^\may$ and $\rvir$} 

Here we prove Theorems~\ref{thm:mayer} and~\ref{thm:virconv}.

\begin{proof}[Proof of Theorem~\ref{thm:mayer}]
	If the activity satisfies Eq.~\eqref{eq:maycrit}, then $z<\rmay_\Lambda(\beta)$, 
	see~\cite[Theorem 2.1]{pog-ue}, and the lower bounds follow just as in 
	 Eq.~\eqref{eq:musat-lb}.
	For the upper bound, we use a  result from~\cite{penrose}: for all $k \in \N$, 
	\begin{equation*}
		\rmay_\Lambda(\beta) \leq \Bigl( \frac{k \exp(- \beta e_\infty)}{(k-1) |b_{k,\Lambda}(\beta)|} \Bigr)^{1/(k-1)}.
	\end{equation*} 
	Here $b_{k,\Lambda}(\beta)$ are the coefficients of the finite volume pressure-density series. They converge to $b_k(\beta)$ as $|\Lambda|\to \infty$, whence 
	\begin{equation*}
		\limsup_{|\Lambda|\to \infty} \rmay_\Lambda(\beta) 
			\leq \Bigl( \frac{k \exp(- \beta e_\infty)}{(k-1) |b_{k}(\beta)|} \Bigr)^{1/(k-1)}.
	\end{equation*}
	We deduce from Theorem~\ref{thm:maylt} that for every $k\in \N$, 
	\begin{equation*}
		\limsup_{\beta \to \infty} \limsup_{|\Lambda|\to \infty}\beta^{-1}\log \rmay_\Lambda(\beta)
		\leq \frac{- e_\infty + E_k }{k-1}.
	\end{equation*}
	We conclude by letting $k\to \infty$ in the upper bound. 
\end{proof}
 
For the proof of Theorem~\ref{thm:virconv}, we
 start with the lower bound on the density of saturated gas $\rho_\sat(\beta)$ 
and the density $\rho^\may(\beta)$ 
delimiting the physical parameter region covered by the Mayer series. 

\begin{proof}[Proof of Eq.~\eqref{eq:vir-a}]
	We observe that $\rho_\sat(\beta) \geq \rho^\may(\beta) \geq \rho(\beta,\mu)$ 
	for all $\beta$ and all $\mu \leq \mu_\sat(\beta)$. By Eq.~\eqref{eq:musat-lb}  
	and Theorem~\ref{thm:cross-over}, it follows that 
	for every $\mu<e_\infty$, 
	\begin{equation*}
		\liminf_{\beta \to \infty} \beta^{-1} \log \rho_\sat(\beta) 
		 \geq \liminf_{\beta \to \infty} \beta^{-1} \log \rho^\may(\beta)
		 \geq  \sup_{k\in \N} (k \mu - E_k).
	\end{equation*}
	Noting that 
	\begin{equation*}
		\sup_{\mu<e_\infty} \sup_{k \in \N} (k \mu - E_k) 
		= \sup_{k \in \N} \sup_{\mu<e_\infty} (k \mu - E_k) = \sup_{k\in \N} (k e_\infty - E_k) = - \nu^*,
	\end{equation*}
	we deduce Eq.~\eqref{eq:vir-a}.
\end{proof}

\begin{proof}[Proof of Eq.~\eqref{eq:vir-b}] 
	Write the pressure-density  series  as $\beta p = \rho + \sum_{n\geq 2} c_n(\beta) \rho^n$. 
	We start from the contour integral, see~\cite{leb-pen} or~\cite[Chapter 4.3]{ruelle-book}, 
	\begin{equation*}
		c_n(\beta) = \frac{\beta^{-1}}{2\pi i} \oint_C \frac{\dd z}{ n z \rho(z)^{n-1}}.
	\end{equation*}
	Here the density $\rho(z) = \sum_{n=1}^\infty n b_n(\beta) z^n$ is extended to
	complex activities $z$, and we integrate on a circle $C$ of 
	radius $\exp(\beta \mu) < \exp(\beta \mu_1)$. For $\beta$ sufficiently large, 
	we know that $\exp(\beta \mu) < \rmay(\beta)$, and we are going to check 
	that $|\rho(z)|>0$ for $|z| = \exp(\beta \mu)$. To this aim we write 
	\begin{equation*}
		|\rho(z)| \geq |z|\, \Bigl( 1 - \sum_{k\geq 2} k |b_k(\beta)|\, |z|^{k-1} \Bigr). 
	\end{equation*}
	For every fixed $K \geq 2$, as $\beta \to \infty$, 
	\begin{align*}
		\sum_{k=2}^K k |b_k(\beta)| (e^{\beta \mu})^{k-1} 
		& = \sum_{k=2}^K k \exp\Bigl( \beta(k-1) \bigl[\mu -  \frac{E_k}{k-1}  + o(1) \bigr]\Bigr) \\
		& \leq \sum_{k=2}^K k \exp \Bigl( \beta (k-1) \bigl(\mu - \mu_1 + o(1)\bigr) \Bigr) \\
		& \leq \const(K,\mu) \exp \Bigl(\beta\bigl (\mu - \mu_1 +o(1)\bigr) \Bigr).    
	\end{align*}
	Furthermore, if $R = \exp( \beta (e_\infty +o(1)))$ is as in Eq.~\eqref{eq:radius}, 
	and $\beta$ sufficiently large so that $\exp( \beta \mu) / R \leq \exp(- \beta \eps)$
	with suitable $\eps>0$, 
	\begin{align*}
		\sum_{k = K+1}^\infty k |b_k(\beta)|\, |z|^{k-1} 
		& \leq \left(\frac{z}{R} \right)^{K}  (e-1) \exp(- \beta e_\infty) \\ 
		& \leq (e-1) \exp\bigl( - \beta (K \eps + e_\infty) \bigr). 
	\end{align*}
	We choose $K \in \N$ such that $K \eps + e_\infty > \mu_1 - \mu$ and combine 
	the previous estimates. We obtain that as $\beta \to \infty$, 
	$\sum_{k \geq 2} k |b_k(\beta)|\, |z|^{k-1}$ is of order at most $\exp( \beta (\mu-\mu_1))$ 
	and, in particular, goes to $0$, so that $\rho(z) \neq 0$. 

	We can plug the lower bound for $|\rho(z)|$ into the contour integral. This yields 
	\begin{equation*}
		|c_n| \leq \frac{\beta^{-1}}{n}\times  \frac{1}{ \bigl[ (1 + o(1)) \exp(\beta \mu)\bigr]^{n-1}}
	\end{equation*}
	whence $\rvir(\beta) \geq (1+o(1)) \exp(\beta \mu)$ and 
	\begin{equation*}
		\liminf_{\beta \to \infty} \beta^{-1} \log \rvir(\beta) \geq \mu.
	\end{equation*} 
	This is true for every $\mu <\mu_1 = - \nu_1$, and the inequality~\eqref{eq:vir-b} follows. 
\end{proof}

\begin{prop} \label{prop:rhozero}
	Under the assumptions of Theorem~\ref{thm:virconv}, suppose that in addition 
	$\mu_1 = E_2 <e_\infty$. 
	Then, for suitable $\beta_1, C_1, \eps>0$ and all $\beta \geq \beta_1$, 
	the equation 
	\begin{equation*}
		\frac{\dd \rho}{\dd z} (z) =1+ \sum_{k=2}^\infty k^2  b_k(\beta) z^{k-1} =0 
	\end{equation*}
	has a solution $z_0(\beta)$ at distance $\leq C_1 \exp( -\eps \beta)/ b_2(\beta)$ 
	of $- 1/(4 b_2(\beta))$, and $\rho(z_0(\beta)) \neq 0$. 
	If $\mu_1 = E_p/(p-1) < e_\infty$ for some $p \geq 3$, a similar statement holds 
	with $-1/4 b_2(\beta)$ replaced with one of the roots 
	of the equation $1 + p^2 b_p(\beta) z^p =0$.
\end{prop}

\begin{proof}[Proof of Eq.~\eqref{eq:vir-c}] 
	For sufficiently small $\rho$, the density-activity relation can be inverted: 
	there is a function $\zeta(\rho)$, analytic in a domain containing $0$,
	such that for small $z$, $\zeta(\rho(z))=z$, and the restriction of $\zeta$ 
	to some neighborhood of $0$ is  injective. 
	 The virial series is given by the composition 
	\begin{equation*}
		P(\rho) = \sum_{n=1}^\infty c_n \rho^n 
			= \sum_{n=1}^\infty b_n \Bigl(\zeta(\rho)\Bigr)^n.
	\end{equation*}
	For sufficiently small  $z$, we have $\zeta(\rho(z)) = z$ and 
	\begin{equation} \label{eq:Gr}
		P'(\rho(z)) = \frac{\rho(z)}{z \rho'(z)}. 
	\end{equation}
	The relation extends by analyticity to every domain $D$ 
	such that $\rho$ is analytic in $z \in D$ and $P$ is analytic in 
	$\rho(D)$. Now, from Prop.~\ref{prop:rhozero}, we know that 
	$\rho'(z_0) /\rho(z_0)= 0$ with $0<|z_0| <\rmay$. Eq.~\eqref{eq:Gr} 
	cannot be true at $z = z_0$. Let $D$ be an open disk centered at $0$ with radius 
	$> |z_0|$. The function $P$ cannot be analytic in all of $\rho(D)$, 
	hence there must be some $z \in D$ such that 
	\begin{equation} \label{eq:boundrho}
		\rvir \leq |\rho(z)| \leq \sum_{k=1}^\infty k |b_k(\beta) z^k|.
	\end{equation}
	Let $\delta>0$ small enough so that $\mu_1+\delta<e_\infty$. 
	For sufficiently large $\beta$, we may choose $|z| \leq \exp( \beta (\mu_1 + \delta))$ 
	in Eq.~\eqref{eq:boundrho}. The usual procedure shows that 
	\begin{equation*}
		\limsup_{\beta \to \infty} \beta^{-1} \log \rvir(\beta) 
			\leq \sup_{k\in\N}\bigl(k (\mu_1+\delta) - E_k\bigr).
	\end{equation*}
	Noting that $\sup_k (k\mu - E_k)$ is locally bounded and convex, hence continuous, 
	in $\mu < e_\infty$, we can let 
	 $\delta \searrow 0$, which yields  
	\begin{equation*}
		\limsup_{\beta \to \infty} \beta^{-1} \log \rvir(\beta) 
			\leq \sup_{k\in\N} (k \mu_1 - E_k) = - \nu_1. \qedhere
	\end{equation*} 
\end{proof}

\begin{proof}[Proof of Prop.~\ref{prop:rhozero}]
	The idea is to use an implicit function theorem, perturbing around $\exp(-\beta) = 0$. 
	We give the proof for $\mu_1 = E_2 <\infty$. The proof for $\mu_1 = E_p/(p-1)$, 
	$p \geq 3$, is similar. 
	For sufficiently large $\beta$, $b_2(\beta) >0$. 
	Let 
	\begin{equation*}
		a_k(\beta):= k^2 \frac{b_k(\beta)}{b_2(\beta)^{k-1}}, \quad 
		\hat z:= b_2(\beta) z. 
	\end{equation*}
	The equation to be solved becomes 
	\begin{equation*}
		1 + 4 \hat z + \sum_{k \geq 3} a_k(\beta) \hat z^{k-1} = 0.  
	\end{equation*}  
	By assumption, for suitable $\Delta >0$ and all $k\geq 3$, 
	$E_k \geq \mu_1 + (k-1) \Delta$, and $\mu_1 = E_2 < e_\infty$. 
	Therefore, for every fixed $k\geq 3$, 
	as $\beta \to \infty$,   
	\begin{equation*}
		|a_k(\beta)| \leq  k^2 \exp\bigl(- \beta(k-1) (\Delta +o(1))\bigr) \to 0.  
	\end{equation*}
	Moreover, for every $s>0$ and $K \in \N$,  
	\begin{align*}
		\sum_{k =K+1}^\infty k^{1+s} \left| \frac{b_k(\beta)}{b_2(\beta)^{k-1}} \right| 
			& \leq \Bigl( \sup_{k \geq K} (k+1)^s (b_2(\beta) R)^{-k} \Bigr) (e-1) e^{-\beta e_\infty}.  
	\end{align*}
	with $R =\exp( \beta (e_\infty +o(1)))$ as in Eq.~\eqref{eq:radius}.  Thus 
	\begin{equation*}
		b_2(\beta) R = \exp\bigl(\beta (e_\infty - E_2 +o(1))\bigr) \to \infty.
	\end{equation*}
	It follows that for every $s>0$, and suitable $\eps_s >0$, as $\beta \to \infty$, 
	\begin{equation*}
		\sum_{k=3}^\infty k^s |a_k(\beta)| = O\bigl(\exp(- \eps_s \beta)\bigr).
	\end{equation*}
	Now let $X$ be the Banach space of sequences $(a_k)_{k \geq 3}$ 
	with weighted norm $||\vect{a}||:= \sum_k k^s |a_k|$. Set 
	\begin{equation*}
		F(\vect{a},\hat z):= 1+ 4 \hat z + \sum_{k \geq 3} a_k \hat z^{k-1}.
	\end{equation*}
	The equation $F(\vect{0},\hat z) =0$ has the unique solution 
	$\hat z_0 = - 1/4$, and in a neighborhood of $(\vect{0},\hat z_0)$, 
	for suitable choice of $s$, 
	$F$ is continuously Fr{\'e}chet-differentiable. Moreover 
	$\partial_{\hat z} F(\vect{0},\hat z) = 4 \neq 0$.  As a consequence, we 
	can apply a Banach space implicit function theorem~\cite[Chap. 4]{zeidler}. It follows 
	in particular that as $||\vect{a}||\to 0$, 
	the equation $F(\vect{a},\hat z) =0$ has a solution $\hat z(\vect{a}) = \hat z_0 + O(||\vect{a}||)$. Applying this to $\vect{a}(\beta)= (a_k(\beta))_{k\in\N}$, we obtain the solution 
	$b_2(\beta) z_0(\beta) = - 1/4 + O(e^{-\beta \eps_s})$. 
	
	For the density, we observe that 
	\begin{equation*}
		\rho(z_0(\beta)) = z_0(\beta) \Bigl( 1+ 2 b_2(\beta) z_0(\beta) + O(e^{-\eps_s \beta}) \Bigr) = - \frac{1}{2} z_0(\beta) (1+ O(e^{-\eps_s \beta})). 
	\end{equation*}
	It follows that for sufficiently large $\beta$, $\rho(z_0(\beta)) \neq 0$. 
\end{proof}

\appendix

\section{Two auxiliary variational problems} \label{app:variational}

Throughout this section we assume that $v$ is a stable pair potential with attractive tail. 
Consider the following two variational problems 
\begin{alignat*}{2}
	\nu(\mu)&:= \inf_{k\in \N} (E_k - k \mu),&\quad \mu &\leq e_\infty,\\
	\mu(\nu)&:= \inf_{k \in \N} \frac{E_k - \nu}{k},& \quad \nu &>0. 
\end{alignat*}
The first variational problem appears in Theorem~\ref{thm:cross-over}, and minimizers 
$k(\mu)$ correspond to the favored size of molecules in the gas phase 
as $\beta \to \infty$ at fixed $\mu$.  
The second problem appears in Eq.~\eqref{eq:free-energy} and, as shown in~\cite{jkm},
minimizers $k(\nu)$ correspond to favored cluster or molecule sizes as $\beta \to \infty$ 
and $\rho \to 0$ along $\rho = \exp(-\beta \nu)$, at fixed $\nu$. 

 Recall that 
for potentials with an attractive tail, $\nu^*:= \inf_k (E_k - k e_\infty) >0$. 

\begin{lemma}[Concavity, monotonicity and equivalence] \label{lem:aux1}
	Let $v$ be a stable pair interaction with attractive tail. Then: 
	\begin{enumerate}
		\item The function $\mu \mapsto \nu(\mu)$ is strictly decreasing, piecewise affine 
				 and 
			concave in $\mu \in(- \infty,e_\infty]$. The function $\nu\mapsto \mu(\nu)$ 
			is decreasing, piecewise affine 
			 and concave in $\nu \in [0,\infty)$; it is strictly decreasing  
			in $\nu \in [\nu^*,\infty)$ and equals $\mu(\nu) = e_\infty$ 
			for $\nu \leq \nu^*$. 
 		\item  
			For $\mu\leq e_\infty$ and $\nu\geq \nu^*$, 	
			$\nu = \nu(\mu)$ if and only if $\mu(\nu)= \nu$.			
	\end{enumerate}
\end{lemma}
The reciprocity of $\mu(\nu)$ and $\nu(\mu)$ 
is analogous to the equivalence of the grand-canonical and the constant pressure ensembles. 
Indeed, the pressure and the Gibbs energy (per particle) are both obtained as Legendre transforms of the free energy,
one with respect to the density, the other with respect to the volume per particle,
\begin{equation*}
	p(\beta,\mu) = \sup_{\rho} \bigl( \mu \rho - f(\beta,\rho)\bigr),\quad 
	g(\beta,p) = \inf_{v} \bigl( \tilde f(\beta,v) + p v\bigr),
\end{equation*}
with $\tilde f(\beta,v) = v f(\beta,v^{-1})$ the free energy per particle. Equivalence of ensembles here means that $p(\beta,\cdot)$ and $g(\beta,\cdot)$ are reciprocal: the Gibbs energy is the same as the chemical potential. 

Similarly, $\mu(\nu)$ looks like a Legendre transform of $k \mapsto E_k$ with respect to $k$,  
while $\nu(\mu)$ looks like a Legendre transform of $E_k/k$ with respect to $1/k$, which should be compared with  
the relations $v = 1/\rho$, $\tilde f(\beta,v) = f(\beta,\rho) /\rho$. 

\begin{proof}[Proof of Lemma~\ref{lem:aux1}]
	1. The statement for the function $\mu(\nu)$ was proven in~\cite{ckms,jkm}. 
	For $\nu(\mu)$, we note that it
	 is the infimum of a family of decreasing, affine functions and therefore 
	concave and decreasing. Moreover it is almost everywhere differentiable, with 
	derivative $- k(\mu)$, the minimizer of $E_k - k \mu$. In particular 
	$k(\mu) \geq 1$, hence $\nu(\mu)$ is strictly decreasing. 

	2.  We prove ``$\Rightarrow$''. The proof of the converse is similar. Thus let 
	$\mu \leq e_\infty$ and $\nu = \inf_k(E_k - k\mu)$.  Clearly, $\nu \leq  \inf_k (E_k - k e_\infty) = \nu^*$, and for every $k \in \N$, 
	\begin{equation*}
		\nu \leq E_k - k\mu \ \Rightarrow \mu \leq \frac{E_k - \nu}{k}, 
	\end{equation*}
	whence $\mu \leq \mu(\nu)$. On the other hand, if $\mu<e_\infty$, then 
	$E_k - k \mu \geq \nu^* + k (e_\infty - \mu) \to \infty$ 
	as $k\to \infty$, so there must be a finite $k$ such that $\nu = E_k - k\mu$.
	It follows that $\mu = (E_k-\nu)/k \geq \mu(\nu)$, whence $\mu = \mu(\nu)$. 
	If $\mu = e_\infty$, then $\nu = \nu^*$ and the claim follows from the general 
	inequality $\mu(\nu) \leq e_\infty$. 
\end{proof}

\begin{lemma}[Comparison of thresholds] \label{lem:thresholds}
	Let 
	\begin{equation*}
		\mu_1:= \inf_{k \geq 2} \frac{E_k}{k-1}, \quad \nu_1 := - \mu_1. 
	\end{equation*}
	Then 
	\begin{itemize}
		\item either $\mu_1= e_\infty$ and $\nu^*=- e_\infty = \nu_1$,
		\item or $\mu_1 <e_\infty$ and $\nu ^*  < - e_\infty < \nu_1$ .
	\end{itemize}
\end{lemma}

\begin{proof}
	Lemma~\ref{lem:aux1} implies the general bounds $\mu_1 \leq e_\infty$ and 
	$\nu_1 \geq \nu^*$. Moreover, by definition, $\nu^*  \leq E_1 - e_\infty = - e_\infty$ and 
	\begin{equation*}
		\nu_1 = \sup_k \frac{E_k}{1-k} \geq \lim_{k \to \infty} \frac{E_k}{1-k} = - e_\infty
	\end{equation*}
	so that $\nu^* \leq - e_\infty \leq \nu_1$. 
	If in addition $\mu_1 = e_\infty$, then $\nu_1 = - e_\infty$ and for all $k \in \N$, 
	$E_k \geq (k-1) e_\infty$ from which we get $\nu^* = \inf_k( E_k - ke_\infty) \geq - e_\infty$. Since in any case $\nu^* \leq e_\infty$, we get $\nu^* = e_\infty$. 

	If $\mu_1 <e_\infty$, then $\nu_1 >-e_\infty$ and there is a $p \in \N$ 
	such that $\mu_1 = E_p/(p-1) < e_\infty$. It follows that 
	\begin{equation*}
	 	\nu^* \leq E_p - pe_\infty = (p-1) (\mu_1 - e_\infty) -  e_\infty < - e_\infty. \qedhere
	\end{equation*}
\end{proof}

\begin{lemma}[``Phase'' diagram] \label{lem:phases}
	\begin{enumerate}
		\item 	For $\mu<\mu_1$, $E_k-k \mu$ has the unique minimizer $k(\mu) =1$. 
	Similarly, for $\nu>\nu_1$, $(E_k - \nu)/k$ has the unique minimizer $k(\nu) =1$. 
		\item For $\mu_1 < \mu < e_\infty$, every minimizer is finite and larger 
		or equal to $2$; similarly for $\nu^*< \nu <\nu_1$. 	
		\item For $\nu<\nu^*$, $(E_k-\nu)/k$ has no finite minimizer. 
	\end{enumerate}
\end{lemma}

\begin{proof}
	1. By definition, $\mu < \mu_1$ if and only if for all $k \geq 2$, 
	$(E_1 - \mu)/1= - \mu < E_k - k \mu$. Thus for $\mu <\mu_1$, $E_k - k \mu$ 
	has the unique minimizer $k(\mu) =1$. The statement on $(E_k -\nu)/k$ is proven 
	in an analogous way. 

	2. For $\mu <e_\infty$, $E_k - k \mu \geq k (e_\infty-\mu) \to \infty$ as $k\to \infty$, 
	thus $(E_k-k\mu)_{k\in \N}$ reaches its minimum at finite values of $k$. If $k=1$ 
	was a minimizer, we would have $- \mu \leq E_k - k \mu$ for all $k \geq 2$, whence 
	$\mu \leq \mu_1$. Therefore when $\mu >\mu_1$, every minimizer $k(\mu)$ is larger or equal to two.
	The proof for the statement on $(E_k - \nu)/k$ is similar. 

	3. By definition, if $\nu < \nu^*$, then $\nu < E_k - k e_\infty$ for all $k \in \N$, 
	thus $(E_k - \nu)/k > e_\infty$ for all $k$. It follows that 
	$\inf_k (E_k - \nu)/k = e_\infty$ and there is no finite minimizer. 
\end{proof}
%
%

\section{Free energy at low temperature and low density} \label{app:free-energy}

Here we give a sketch of the proof of~\eqref{eq:free-energy}. 
The primary aim is to show that $\rho_0$ can be chosen indeed of the order of the preferred ground state density, $\rho_0 \approx 1/a^d$. For the sake of completeness, we also make a remark on how Eq.~\eqref{eq:free-energy} should be modified for potentials without attractive tail. 

\paragraph{Potentials with attractive tail} 
Let $Z_k^\cl(\beta)$ be the cluster partition function from~\eqref{eq:zk} above. 
Then~\cite[Lemma 3.1]{jkm}
\begin{equation*}
	Z_\Lambda(\beta,N) \leq \sum_{\sum_1^N {kN_k = N}}\ \prod_{k=1}^N \frac{ (|\Lambda| Z_k^\cl(\beta))^{N_k}} {N_k!}.
 \end{equation*}
The sum is over integers $N_1,\ldots,N_N\in\N_0$ such that $\sum_k k N_k = N$.
The integers describe a partition of the $N$ particles into clusters, i.e., groups of particles close in space. Using that for suitable $c>0$ and all $\beta>0$ and 
$k\in \N$, 
\begin{equation*} 
	 Z_k^\cl(\beta) \leq \exp( - \beta E_k) \exp( c k), 
\end{equation*}
\cite[Lemma 4.3]{jkm}
we deduce that $-\beta f(\beta,\rho)$ is upper 
bounded by the supremum of 
\begin{equation*}
	c \rho - \beta \Bigl( \sum_{k\in \N} \rho_k E_k + (\rho-\sum_{k\in \N} k\rho_k) e_\infty \Bigr) 
		+  \sum_{k\in \N} \rho_k (1 - \log \rho_k). 
\end{equation*}
over all $(\rho_k)_{k\in \N} \in [0,\infty)^\N$ such that $\sum_1^\infty k \rho_k \leq \rho$ 
(think $\rho_k = N_k/|\Lambda|$). Next, we observe that the mixing entropy 
can be bounded as $\sum_k \rho_k \log (\rho_k/\rho) \geq - 2\rho$, for all $\rho>0$ 
and all admissible $(\rho_k)$, see~\cite[Lemma  4.2]{jkm}. Therefore we obtain 
\begin{multline*}	
	- \beta f(\beta,\rho) \leq (c+3) \rho \\
	- \beta \inf\Bigl\lbrace (\rho-\sum_{k\in\N} 
			k \rho_k) e_\infty + \sum_{k\in\N} \rho_k (E_k + \beta ^{-1}\log \rho) 
			\,\Big|\, \sum_k k \rho_k \leq \rho \Bigr \rbrace,  
\end{multline*}
whence 
\begin{equation}  \label{eq:flb}
	f(\beta,\rho) \geq - (c+3) \beta^{-1} \rho + \rho \inf_{k\in \N} \frac{E_k + \beta^{-1} \log \rho}{k},
\end{equation}
for all $\beta>0$ and all $\rho>0$. 

It remains to obtain an upper bound for the free energy, or 
a lower bound for the partition function. 
Consider first the case $\nu^*>0$ and $\rho \geq \exp(- \beta \nu^*)$. In this case 
the infimum in Eq.~\eqref{eq:flb} equals $e_\infty$. 
 We lower bound the partition function by integrating only over a small neighborhood of the $N$-particle ground state, 
and deduce $-  f(\beta,\rho) \geq e_\infty - C \beta^{-1} \log \beta$ for suitable $C$ 
and sufficiently large $\beta$. Note that this is possible if $\rho$ is smaller than the density of the ground state, thus  $\rho < 1/a^d$ is sufficient.

Next, consider the case $\rho = \exp(- \beta \nu)< \exp( - \beta \nu^*)$. In this 
case $(E_k - \nu)/k$ has a finite minimizer $k= k(\nu) \in \N$. We lower bound the partition function for a cube $\Lambda = [0,L]^d$ and $N \in k \N$ particles as follows: we split the cube into $M$ small cubes (``cells'')
with side length of the order $ak^{1/d}$ and mutual distance $R$.  
 Here $R$ is of the order of the potential range, and $a k^{1/d}$ is large enough so that a $k$-particle ground state fits into the small cube, as in Assumption~\ref{ass:uniform}.  We can place 
approximately $M =|\Lambda| / (a k^{1/d}+ R)^d$ small cubes in that way. 
We consider configurations in which particles form clusters of size $k$ close to their ground state, 
such that each cluster fits completely into a small cube, and there is at most one cluster per 
cell. 

 We refer the reader to~\cite{jkm} for the details and content ourselves with the following remark: the procedure works provided the number $M$ of available cells 
is larger than $N/k$. This gives the condition 
\begin{equation*}
	\rho < ( a+ R k^{-1/d})^{-d}.
\end{equation*}
Hence if we choose $\beta$ large enough so that $\exp(- \beta \nu^*) \leq (a +R)^{-d}$,
the condition is certainly fulfilled for every $\rho \leq \exp(-\beta \nu^*)$. 

Remember that for $\rho \geq \exp(-\beta \nu^*)$, we are in the first case considered above 
and we only need $\rho \leq 1/a^d\approx$ ground state density. Therefore, in the end, all we need is the 
condition $\rho <\rho_0$ with $\rho_0$ of the order of $1/a^d$. 

\paragraph{Potentials without attractive tail} 
If $v$ has no attractive tail,  we might have $\nu^* =0$, and ground states are not necessarily connected. Set 
$E_1^\cl = E_1 = 0$ and  
\begin{equation*}
	E_k^\cl:= \inf \bigl\lbrace U(x_1,\ldots,x_k)\mid (x_1,\ldots,x_k) \in (\R^d)^k\ R\text{-connected} \bigr\rbrace \geq E_k.
\end{equation*}
The lower bound~\eqref{eq:flb} is still true, but can be improved by replacing $E_k$ by $E_k^\cl$. 
In fact, indices $k$ with $E_k^\cl >E_k$ can be dropped altogether: 
if $E_k^\cl>E_k$, then for suitable $r\geq 2$, $k_1+\cdots + k_r =k$, 
\begin{equation} \label{eq:deco}
	E_k^\cl >E_k=  E_{k_1}^\cl + \cdots + E_{k_r}^\cl 
\end{equation}
Suppose that for some $\nu >0$, 
$(E_{k_i}^\cl - \nu ) /k_i \geq (E_k^\cl - \nu)/k$ for all $i$. Then  
\begin{equation*}
	\sum_1^r E_{k_i}^\cl \geq \sum_1^r \bigl( \nu + k_i \frac{E_k^\cl - \nu}{k} \bigr) 
			= (r-1) \nu + E_k^\cl > E_k^\cl,
\end{equation*}
contradicting~\eqref{eq:deco}. 
Thus 
\begin{equation*}
	\inf_{k\in \N} \frac{E_k^\cl - \nu}{k} = \inf\bigl\lbrace \frac{E_k - \nu}{k} \mid k \in \N,\  E_k = E_k^\cl \bigr \rbrace
\end{equation*}
is the appropriate auxiliary variational problem to be substituted into Eq.~\eqref{eq:free-energy}. 
The density $\rho_0$ can be chosen of the order of $(a+R)^{-1/d}$. 

\paragraph{Non-negative potentials}
When $v \geq 0$, the situation becomes particularly simple: 
we have $\nu^*=0$ and for all $\nu>0$, 
\begin{equation*}
	\inf_{k\in \N} \frac{E_k - \nu}{k} = \inf_{k\in \N} \frac{E_k^\cl - \nu}{k} = - \nu
\end{equation*}
and $k(\nu) = 1$ is the unique minimizer. Eq.~\eqref{eq:free-energy} is replaced by the following: 
For sufficiently low temperature and density $\rho$ smaller than or of the order of $1/R^d$, 
\begin{equation*}
	\bigl| f(\beta,\rho) - \beta^{-1} \rho \log \rho \bigr| \leq C \rho \beta^{-1} \log \beta
\end{equation*}
and we recognize the free energy $\beta^{-1} \rho (\log \rho - 1)$ of an ideal gas.

\section{Cluster expansion in the canonical ensemble} \label{app:clexp}

The virial expansion~\eqref{eq:virial-fe} can be derived directly with the help 
of a cluster expansion in the canonical ensemble; this was recently done in~\cite{pulv-tsag}.
The aim of this appendix is to  complement Sect.~\ref{sec:vircoeff} and explain how Eq.~\eqref{eq:dnbn} (without the exact formula for $a(\vect{m})$) is obtained with the approach from~\cite{pulv-tsag}. 


The starting point is an expression of the canonical partition function as a sum over set partitions $\{X_1,\ldots,X_r\}$, $r \in \N$, of the particle label set $\{1,\ldots,N\}$: 
\begin{equation*}
	Z_\Lambda(\beta,N) = \frac{|\Lambda|^N}{N!} \sum_{\{X_1,\ldots,X_r\}} \zeta_\Lambda(X_1) \cdots \zeta_\Lambda(X_r). 
\end{equation*}
Monomers ($|X|=1$) have activity $1$, sets with higher cardinality have activity 
\begin{equation*}
	|X|= k \geq 2:\ \zeta_\Lambda(X) = \frac{1}{|\Lambda|^k} 
		\sum_{\gamma \in 	\mathcal{G}_\mathrm{c}(k)} \int_{\Lambda^k} \prod_{(ij) \in \gamma} f_{ij}(\vect{x}) \dd x_1 \cdots \dd x_k,  
\end{equation*}
with $\mathcal{G}_\mathrm{c}(k)$ the set of connected graphs with vertices $1,\ldots,k$. 
Note that as $|\Lambda|\to \infty$, for every fixed $k$ and $\beta$, the activity is related to 
the Mayer coefficients as follows: 
\begin{equation*}
	|X| =k:\quad \zeta_\Lambda(X) \sim \frac{k!}{|\Lambda|^{k-1}}\,  b_k(\beta) 
		=: \frac{B_k(\beta)}{|\Lambda|^{k-1}}
\end{equation*}
The formalism of cluster expansions for polymer partition functions gives 
\begin{equation*}
	\log Z_\Lambda(\beta,N) = \log \frac{|\Lambda|^N}{N!} + 
		\sum_{r \geq 1} \frac{1}{r!} \sum_{\stackrel{X = (X_1,\ldots,X_r)}{\conn,\ X_i \in \Gamma_N}} 
				n(X)\, \zeta_\Lambda(X_1) \cdots \zeta_\Lambda(X_r). 
\end{equation*} 
Here $\Gamma_N$ is the collection of subsets of $\{1,\ldots,N\}$ of cardinality at least $2$, and connectedness 
and $n(X)$ are defined as follows: 
 With $X = (X_1,\ldots,X_r) \in \Gamma_N^r$ we associate the graph $G(X)$ with vertices $1,\ldots,r$ and edges 
$\{i,j\}$, $i \neq j$, $X_i \cap X_j \neq \emptyset$. The polymer $X$ is called connected if the graph of overlaps $G(X)$ is, and $n(X) \in \N_0$ is the index of $G(X)$, i.e.,
 $n(X) = n_+(X) - n_-(X)$ with $n_\pm(X)$ the number of connected subgraphs 
of $G(X)$ with an even (odd) number of edges. 

The $N$-dependence in the summation index is slightly inconvenient. We remove it by 
exploiting the invariance with respect to particle relabeling,
\begin{multline*}
	\log Z_\Lambda(\beta,N) = \log \frac{|\Lambda|^N}{N!} + 
		\sum_{n=2}^N \binom{N}{n}
		\sum_{r = 1}^\infty  \frac{1}{r!} \sum_{\stackrel{X = (X_1,\ldots,X_r)}{\conn,\ X_i \in \Gamma_n}}
	n(X) \\
	\times  \zeta_\Lambda(X_1) \cdots \zeta_\Lambda(X_r) \ 
		\mathbf{1}\bigl(\cup_1^r X_i = \{1,\ldots,n\}\bigr).
\end{multline*}
In the thermodynamic limit $N,|\Lambda|\to \infty$, for each cluster $X=(X_1,\ldots,X_r)$ in the sum, 
\begin{multline*}
	\frac{1}{|\Lambda|} \binom{N}{n}  \zeta_\Lambda(X_1) \cdots \zeta_\Lambda(X_r) 
	 \sim \frac{1}{n!} \frac{N^n}{|\Lambda| ^{ 1+ \sum_1^r(k_i-1)}} \\
			\times  B_{k_1}(\beta)\cdots B_{k_r}(\beta), \quad 	k_i = |X_i|.  
\end{multline*}
This goes to zero unless $1+ \sum_1^{r}(k_i - 1) = n$ (note that ``$\geq$'' for every cluster $X$).  
When this condition is satisfied, the components are necessarily distinct, $X_i \neq X_j$, 
and  the overlap graph $G(X)$ is necessarily a 
Husimi graph, i.e., a graph whose doubly connected components are complete graphs. 
Using the fact that the index of the complete graph on $v$ vertices is 
$(-1)^{v-1} (v-1)!$, one finds that 
\begin{equation*}
	n(X) = (-1)^{r-1} \prod_{i=1}^j (v_i-1)! 
\end{equation*} 
with $j$ the number of doubly connected components of $G(X)$ and $v_1,\ldots,v_j$ 
their respective sizes; thus $v_1+ \cdots + v_j = r$. 

Assuming we can exchange summation and thermodynamic limits, we obtain 
\begin{equation} \label{eq:fbrho}
\begin{aligned}
	- \beta f(\beta,\rho)&=- \rho (\log \rho -1) + \sum_{n=2}^\infty \frac{\rho^n}{n!} B_n(\beta) \\
		& \quad+ \sum_{n=2}^\infty \frac{\rho^n}{n!} \sum_{r\geq 2}   \ \ 
		\sideset{}{^{(n)}} \sum_{X= \{X_1,\ldots,X_r\}} \ 
			n(X) \prod_{i=1}^r B_{|X_i|}(\beta). 
\end{aligned}
\end{equation}
The sum $\sum^{(n)}$ is over collections of subsets $\{X_1,\ldots,X_r\}$  
with $(X_1,\ldots,X_r)$ connected, $|X_i|\geq 2$ for all $i$, and   
such that $\cup_1^r X_i = \{1,\ldots,n\}$ and $n = 1 + \sum_1^r (|X_i|- 1)$. 
Setting 
\begin{equation*}
	m_j:= \bigl | \{ i\mid 1 \leq i \leq r,\ |X_i|=j \} \bigr|, 
\end{equation*}
we have 
\begin{equation} \label{eq:m}
	n-1 = \sum_{i=1}^r (|X_i|-1) = \sum_{j=2}^n (j-1) m_j, \ r = \sum_{j=2}^n m_j,
\end{equation} 
and we obtain the expansion~\eqref{eq:dnbn} with the information $a(\vect{m}) >0$. 
Additional combinatorial steps would be needed to obtain the formula for the $a(\vect{m})$'s, 
but the information that they are strictly positive is all that is needed for 
proofs of Propositions~\ref{prop:virsin} and~\ref{prop:vircon}.

\paragraph{Acknowledgments}
  The author gratefully acknowledges very useful discussion with D. Ueltschi, D. Tsagkarogiannis, E. Presutti, E. Pulvirenti, 
  B. Metzger and W. K{\"o}nig,  and also thanks E. Presutti for hospitality during a visit at the 
  University of Rome ``Tor 
  Vergata.''
  This work was supported by the DFG Forschergruppe 718 ``Analysis and stochastics in complex physical systems.''

 \providecommand{\bysame}{\leavevmode\hbox to3em{\hrulefill}\thinspace}
\providecommand{\MR}{\relax\ifhmode\unskip\space\fi MR }
\providecommand{\MRhref}[2]{%
  \href{http://www.ams.org/mathscinet-getitem?mr=#1}{#2}
}
\providecommand{\href}[2]{#2}


\end{document}